\documentclass[10pt,twoside]{article}
\usepackage{amsmath,amssymb,latexsym,graphicx,hyperref,epstopdf,color,subfigure}
\usepackage[noend]{algpseudocode}
\usepackage{algorithm,epstopdf}
\usepackage[capitalise]{cleveref}
\normalsize
\newtheorem{definition}{Definition}

\newtheorem{theorem}{Theorem}
\newtheorem{lemma}{Lemma}

\newtheorem{instance}{Instance}

\newenvironment{proof}{{\sc Proof. }}{\hfill$\Box$\vspace{0.2in}}

\def\mcA{\mathcal{A}}

\def\mcI{\mathcal{I}}
\def\mcJ{\mathcal{J}}
\def\mcM{\mathcal{M}}

\title{Randomized algorithms for fully online multiprocessor scheduling with testing}
\author{
	Mingyang~Gong\thanks{Department of Computing Science, University of Alberta.  Edmonton, Alberta T6G 2E8, Canada.
	\texttt{\{mgong4, guohui\}@ualberta.ca}}
\and
	Zhi-Zhong~Chen\thanks{Division of Information System Design, Tokyo Denki University. Saitama 350-0394, Japan.
	\texttt{zzchen@mail.dendai.ac.jp}}
\and
	Guohui~Lin$^*$\thanks{Correspondence author. Email: \texttt{guohui@ualberta.ca}}
\and
	Lusheng~Wang\thanks{Department of Computer Science, City University of Hong Kong.  Hong Kong SAR, China.
	\texttt{cswangl@cityu.edu.hk}}
}
\date{\today}
\begin{document}
\maketitle

\begin{abstract}
We contribute {\em the first} randomized algorithm that is an integration of {\em arbitrarily many} deterministic algorithms
for the fully online multiprocessor scheduling with testing problem.
When there are only two machines,
we show that with two component algorithms its expected competitive ratio is already strictly smaller than
the best proven deterministic competitive ratio lower bound.
Such algorithmic results are rarely seen in the literature.

Multiprocessor scheduling is one of the first combinatorial optimization problems that have received numerous studies.
Recently, several research groups examined its testing variant,
in which each job $J_j$ arrives with an upper bound $u_j$ on the processing time and a testing operation of length $t_j$;
one can choose to execute $J_j$ for $u_j$ time,
or to test $J_j$ for $t_j$ time to obtain the exact processing time $p_j$ followed by immediately executing the job for $p_j$ time.
Our target problem is the fully online version, 
in which the jobs arrive in sequence so that the testing decision needs to be made at the job arrival as well as the designated machine.
We propose an expected $(\sqrt{\varphi + 3} + 1) (\approx 3.1490)$-competitive randomized algorithm
as a {\em non-uniform} probability distribution over arbitrarily many deterministic algorithms,
where $\varphi = \frac {\sqrt{5} + 1}2$ is the Golden ratio.
When there are only two machines, we show that our randomized algorithm based on two deterministic algorithms is already
expected $\frac {3 \varphi + 3 \sqrt{13 - 7\varphi}}4 (\approx 2.1839)$-competitive.
Besides, we use Yao's principle to prove lower bounds of $1.6682$ and $1.6522$ on the expected competitive ratio
for any randomized algorithm at the presence of at least three machines and only two machines, respectively,
and prove a lower bound of $2.2117$ on the competitive ratio for any deterministic algorithm when there are only two machines.

\end{abstract}

\subsubsection*{Acknowledgments.}
This research is supported by the NSERC Canada,
the Grant-in-Aid for Scientific Research of the Ministry of Education, Science, Sports and Culture of Japan, under Grant No. 18K11183,
the National Science Foundation of China (NSFC: 61972329),
and GRF grants for Hong Kong Special Administrative Region, China (CityU 11210119, CityU 11206120, CityU11218821).

\newpage
\section{Introduction}
We study the {\em fully online multiprocessor scheduling with testing}~\cite{DEM18,DEM20,AE20,AE21} problem in this paper,
and study it from the randomized algorithm perspective.
Multiprocessor scheduling~\cite{GJ79} is one of the first well-known NP-hard combinatorial optimization problems,
having received extensive research in the past several decades.

An instance $I$ of multiprocessor scheduling consists of a set of $n$ jobs $\mcJ = \{J_1, J_2, \ldots, J_n\}$,
each to be executed {\em non-preemptively} on one of a set of $m$ parallel identical machines $\mcM = \{M_1, M_2, \ldots, M_m\}$;
and the goal is to minimize the makespan $C_{\max}$, that is, the maximum job completion time.
Different from the classic setting where each job $J_j$ comes with the processing time $p_j$,
in {\em scheduling with testing} each job $J_j$ arrives with an upper bound $u_j$ on the processing time $p_j$ and a testing operation of length $t_j$,
but $p_j$ remains unknown until the job is tested.
The job $J_j$ can either be executed on one of the machines for $u_j$ time
or be tested for $t_j$ time followed by immediately executing for $p_j$ time on the same machine.

If all the jobs arrive at time zero, multiprocessor scheduling with testing is a {\em semi-online} problem,
denoted as $P \mid t_j, 0 \le p_j \le u_j \mid C_{\max}$.
In this paper, we investigate the {\em fully online} problem in which the jobs arrive in sequence
such that the testing decision needs to be made at the job arrival as well as the designated machine for testing and/or executing,
denoted as $P \mid online, t_j, 0 \le p_j \le u_j \mid C_{\max}$.
Apparently, semi-online is a special case of fully online,
and in both cases the scheduler should take advantage of the known information about a job upon its arrival to
decide whether or not to test the job so as to best balance the total time spent on the job due to the unknown processing time.

Given a polynomial time deterministic algorithm for the semi-online or fully online problem,
let $C(I)$ be the makespan produced by the algorithm on an instance $I$ and $C^*(I)$ be the makespan of the optimal {\em offline} schedule, respectively.
The performance of the algorithm is measured by the {\em competitive ratio} defined as $\sup_I \{C(I)/C^*(I)\}$,
where $I$ runs over all instances of the problem,
and the algorithm is said $\sup_I \{C(I)/C^*(I)\}$-{\em competitive}.
Switching to a randomized algorithm, we correspondingly collect its {\em expected} makespan $E[C(I)]$ on the instance $I$ and
the randomized algorithm is said {\em expected} $\sup_I \{E[C(I)]/C^*(I)\}$-competitive.
For online problems, randomized algorithms sometimes can better deal with uncertainties leading to
lower expected competitive ratios than the competitive ratios of the best deterministic algorithms.
We contribute such a randomized algorithm for $P \mid online, t_j, 0 \le p_j \le u_j \mid C_{\max}$,
and furthermore, when there are only two machines,
we show that its expected competitive ratio is strictly smaller than the best proven competitive ratio lower bound of any deterministic algorithm.

We remind the readers that in our problem the job processing is non-preemptive.
In the literature, researchers have also considered {\em preemptive} job processing~\cite{DEM18,DEM20,AE20,AE21},
where any testing or execution operation can be interrupted and resumed later,
or the more restricted {\em test-preemptive} variant~\cite{AE21},
where the testing and execution operations of a tested job are non-preemptive but
the execution operation does not have to follow immediately the testing operation or on the same machine.
Also, our goal is to minimize the makespan $C_{\max}$, that is, the min-max objective;
while another important goal to minimize the total job completion time, or the min-sum objective,
has received much research too~\cite{DEM18,DEM20}.

\subsection{Previous work on fully online case}
There are not too many existing, deterministic or randomized, approximation algorithms for
the fully online problem $P \mid online, t_j, 0 \le p_j \le u_j \mid C_{\max}$.
We first distinguish a special case where all the testing operations have a unit time, i.e., $t_j = 1$ for every job $J_j$,
called the {\em uniform testing case}~\cite{DEM18,DEM20,AE21}, denoted as $P \mid online, t_j = 1, 0 \le p_j \le u_j \mid C_{\max}$.
Note that in the {\em general testing case}, the testing times can be any non-negative values.

When there is only a single machine, the job processing order on the machine is irrelevant to the makespan.
This hints that the fully online and the semi-online problems are the same.
The first set of results is on the semi-online uniform testing problem $P1 \mid t_j = 1, 0 \le p_j \le u_j \mid C_{\max}$,
due to D{\"u}rr et al.~\cite{DEM18,DEM20}.
They proposed to test the job $J_j$ if $u_j \ge \varphi = \frac {\sqrt{5}+1}2$,
or to test it with probability $f(u_j) = \max\left\{0, \frac {u_j(u_j - 1)}{u_j(u_j - 1) + 1}\right\}$,
leading to a deterministic $\varphi$-competitive algorithm and a randomized expected $\frac 43$-competitive algorithm, respectively.
Let $r_j = \frac {u_j}{t_j}$;
Albers and Eckl~\cite{AE20} extended the above two algorithms for the general testing case $P1 \mid t_j, 0 \le p_j \le u_j \mid C_{\max}$,
to test the job $J_j$ if $r_j \ge \varphi$ or to test it with probability $f(r_j)$,
and achieved the same competitive ratio and expected competitive ratio, respectively.
The authors~\cite{DEM18,DEM20,AE20} showed that both algorithms are the best possible,
that is, $\varphi$ is a lower bound on the competitive ratio for any deterministic algorithm
and $\frac 43$ is a lower bound on the expected competitive ratio for any randomized algorithm by Yao's principle~\cite{Yao77}, respectively.

When there are at least two machines, fully online is more general than semi-online.
Albers and Eckl~\cite{AE21} proposed to test the job $J_j$ if $r_j \ge \varphi$ in the list scheduling rule~\cite{Gra66}
that assigns each job to the least loaded machine for processing (possible testing and then executing).
They showed that such an algorithm is $\varphi (2 - \frac 1m)$-competitive and the analysis is tight, where $m$ is the number of machines.
They also showed a lower bound of $2$ on the competitive ratio for any deterministic algorithm even in the uniform testing case
$P \mid online, t_j = 1, 0 \le p_j \le u_j \mid C_{\max}$~\cite{AE21},
and a slightly better lower bound of $2.0953$ for the two-machine general testing case $P2 \mid online, t_j, 0 \le p_j \le u_j \mid C_{\max}$,
i.e., when there are only two machines~\cite{AE21}.

\subsection{Previous work on semi-online case}
See the above reviewed results when there is only a single machine.

When there are at least two machines,
recall that in the semi-online problems, all jobs arrive at time zero,
and therefore the scheduler can take advantage of all the known $u_j$ and $t_j$ values at time zero for not only job testing decision making
but also to form certain job processing orders to better reduce the makespan.
Indeed this is the case in the pioneering so-called SBS algorithm by Albers and Eckl~\cite{AE21} for the general testing case,
which is $3.1016$-competitive (when $m$ tends to infinity).
For the uniform testing case, Albers and Eckl~\cite{AE21} also proposed a $3$-competitive algorithm,
along with a lower bound of $\max\{\varphi, 2 - \frac 1m\}$ on the competitive ratio for any deterministic algorithm.

The above two competitive ratios for the general and the uniform testing cases had been improved to $2.9513$ and $2.8081$ by Gong and Lin~\cite{GL21}, respectively,
and the state-of-the-art ratios are $2.8019$ and $2.5276$ by Gong et al.~\cite{GFL22}, respectively.

For proving the lower bound of $2 - \frac 1m$ on the competitive ratio,
Albers and Eckl~\cite{AE21} presented an instance of $P \mid t_j = 1, 0 \le p_j \le u_j \mid C_{\max}$
that forces any deterministic algorithm to test all the jobs.
This implies that the lower bound is not only on the competitive ratio for any deterministic algorithm,
but also on the expected competitive ratio for any randomized algorithm by Yao's principle~\cite{Yao77}.
Furthermore, the lower bound of $2 - \frac 1m$ is then on the expected competitive ratio for any randomized algorithm
for the fully online uniform testing problem $P \mid online, t_j = 1, 0 \le p_j \le u_j \mid C_{\max}$.

All the above reviewed competitive ratios and lower bounds are plotted in Figure~\ref{fig00}.
We point out that for both the fully online and the semi-online problems at the presence of at least two machines,
there is no existing randomized algorithm, except deterministic algorithms {\em trivially} treated as randomized.

\begin{figure}[htb]
\begin{center}
\unitlength=1pt
\begin{picture}(400, 210)
\put(10, 125){\tiny \it general testing:}
\put(10, 30){\tiny \it uniform testing:}
\put(330, 73){\tiny \it randomized axis}
\put(330, 66){\tiny \it deterministic axis}
\put(10, 70){$\dots$}
\put(30, 71){\circle*{4}}
\put(28, 57){$1$}
\put(30, 70){\vector(1, 0){350}}
\put(30, 72){\vector(1, 0){350}}
\put(211, 70){\circle{4}}
\put(211, 140){\vector(0, -1){65}}
\put(198, 145){$2.8019 - \frac 1m$~\cite{GFL22}}
\put(226, 70){\circle{4}}
\put(226, 120){\vector(0, -1){45}}
\put(213, 125){$2.9513 - \frac 4{3m}$~\cite{GL21}}
\put(242, 70){\circle{4}}
\put(242, 100){\vector(0, -1){25}}
\put(239, 105){$3.1016 (1 - \frac 1{3m})$~\cite{AE21}}
\put(261, 67){$\triangleleft$}
\put(264, 83){\vector(0, -1){10}}
\put(261, 85){$2\varphi$~\cite{AE21}}
\put(246, 70){$\triangleleft$}
\put(249, 160){\vector(0, -1){85}}
\put(240, 165){\bf $3.1490$~(this paper)}
\put(144, 70){$\triangleleft$}
\put(147, 160){\vector(0, -1){85}}
\put(134, 165){\bf $2.1839$~(this paper)}
%
\put(183, 70){\circle*{3}}
\put(183, 10){\vector(0, 1){55}}
\put(180, 0){$2.5276(1 - \frac 1{4m})$~\cite{GFL22}}
\put(211, 70){\circle*{3}}
\put(211, 30){\vector(0, 1){35}}
\put(208, 20){$2.8081 - \frac 2{3m}$~\cite{GL21}}
\put(230, 70){\circle*{3}}
\put(230, 50){\vector(0, 1){15}}
\put(227, 40){$3(1 - \frac 1{3m})$~\cite{AE21}}
%
\put(64, 72){\circle{6}}
\put(64, 120){\vector(0, -1){45}}
\put(61, 125){$\frac 43$~\cite{AE20}}
\put(92, 70){\circle{6}}
\put(92, 100){\vector(0, -1){25}}
\put(88, 105){$\varphi$~\cite{AE20}}
\put(92, 70){\color{red}\circle*{4.5}}
\put(92, 70){\circle*{3}}
\put(92, 50){\vector(0, 1){15}}
\put(88, 40){$\varphi$~\cite{DEM18,DEM20}}
\put(64, 72){\color{red}\circle*{4.5}}
\put(64, 72){\circle*{3}}
\put(64, 30){\vector(0, 1){37}}
\put(58, 20){$\frac 43$~\cite{DEM18,DEM20}}
%
\put(123, 70){\color{red}\circle*{4}}
\put(123, 30){\vector(0, 1){35}}
\put(110, 20){$\max\{\varphi, 2 - \frac 1m\}$~\cite{AE21}}
\put(121, 70){\color{red}$\triangleright$}
\put(123, 10){\vector(0, 1){58}}
\put(110, 5){$2 - \frac 1m$~\cite{AE21}}
\put(127, 67){\color{red}\small $\blacktriangleright$}
\put(130, 50){\vector(0, 1){15}}
\put(131, 40){2~\cite{AE21}}
\put(136, 67){\color{red}$\triangleright$}
\put(138, 83){\vector(0, -1){10}}
\put(125, 85){2.0953~\cite{AE21}}
\put(149, 67){\color{red}$\triangleright$}
\put(151, 103){\vector(0, -1){30}}
\put(138, 105){\bf $2.2117$~(this paper)}
\put(96, 70){\color{red}$\triangleright$}
\put(98, 180){\vector(0, -1){105}}
\put(98, 185){\bf $1.6682$~(this paper)}
\put(93, 70){\color{red}$\triangleright$}
\put(95, 200){\vector(0, -1){125}}
\put(84, 205){\bf $1.6522$~(this paper)}
\end{picture}
\end{center}
\caption{Competitive ratios of the existing deterministic and randomized algorithms for
	the fully online and semi-online multiprocessor scheduling with testing problems.
	On the $x$-axes, the red-colored symbols indicate lower bounds;
	the triangles ($\triangleleft$'s and $\triangleright$'s for general testing and $\blacktriangleright$ for uniform testing)
	refer to the results for the fully online problems,
	while the circles ($\circ$'s for general testing and $\bullet$'s for uniform testing) refer to those for the semi-online problems.
	Between the two $x$-axes, the lower one is for deterministic algorithms, while the upper is for randomized ones;
	the results above the axes are for the general testing case, 	while those below are for the uniform testing case.\label{fig00}} 
\end{figure}
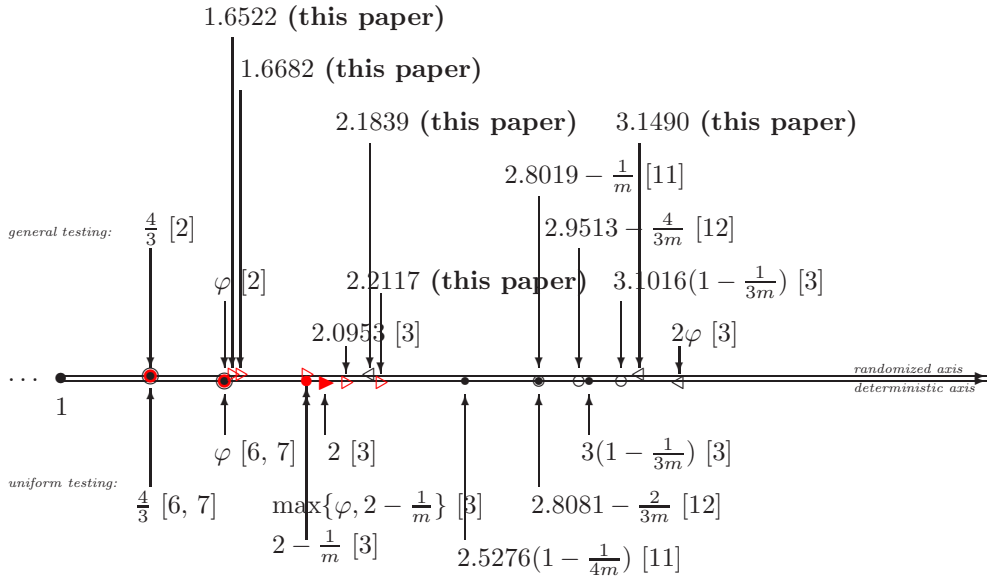

\subsection{Other related work}
We will not review the existing deterministic algorithms for the min-sum objective,
but would like to point out that 
D{\"u}rr et al.~\cite{DEM18,DEM20} presented an expected $1.7453$-competitive randomized algorithm for
the single machine uniform testing problem $P1 \mid t_j = 1, 0 \le p_j \le u_j \mid \sum_j C_j$,
which tests the to-be-tested jobs in a uniform random order.
They also proved a lower bound of $1.6257$ on the expected competitive ratio.
Albers and Eckl~\cite{AE20} designed a randomized algorithm for the general testing case $P1 \mid t_j, 0 \le p_j \le u_j \mid \sum_j C_j$,
in which the job $J_j$ is tested with a rather complex probability function in the ratio $r_j$;
the algorithm is expected $3.3794$-competitive.

We next review the state-of-the-art randomized algorithms for the classic online multiprocessor scheduling problem
$P \mid online \mid C_{\max}$,
in which the jobs arrive in sequence, the processing time $p_j$ is revealed when the job $J_j$ arrives,
and upon arrival the job must be assigned to a machine irrevocably for execution in order to minimize the makespan.
It is well-known that for this problem, the improvement from the competitive ratio $2$,
by the first deterministic algorithm list-scheduling~\cite{Gra66},
to the current best $1.9201$~\cite{FW00} took a long time.

For every online approximation algorithm, one of the most challenging parts in performance analysis is to precisely estimate the optimal offline makespan.
The three most typical lower bounds are 
$\frac 1m \sum_{j = 1}^n p_j$ (the average machine load),
$\max_{j = 1}^n p_j$ (the largest job processing time),
and $p_{[m]} + p_{[m+1]}$ where $p_{[j]}$ is the $j$-th largest processing time among all the jobs.
Albers~\cite{Alb02} proved that if only the above three lower bounds are used,
then the competitive ratio of any deterministic algorithm cannot be smaller than $1.919$ when $m$ tends to infinity.

On the positive side, Albers~\cite{Alb02} presented an expected $1.916$-competitive randomized algorithm,
and in the performance analysis only the above three lower bounds on the optimal offline makespan are used.
In more details, Albers proposed two deterministic algorithms $A_0$ and $A_1$,
and proposed to execute each with probability $0.5$.
The author showed that if the expected makespan is too large compared to the three lower bounds,
that is, both $A_0$ and $A_1$ produce large makespans,
then there are many large jobs present in the instance so that the optimal offline makespan cannot be too small.

Such a randomized algorithm that is a distribution of a constant number of deterministic algorithms is referred to as
a {\em barely randomized algorithm}~\cite{RWS94},
which takes advantage of its good component algorithms, in that it performs well if at least one component algorithm performs well,
and if no component algorithm performs well then the optimal offline makespan is also far away from the lower bounds.
This design idea has been used in approximating many other optimization problems and achieved breakthrough results, for example,
other scheduling problems~\cite{Sei03,FPZ08} and the $k$-server problem~\cite{BCL00}.
In fact, for the online multiprocessor scheduling problem $P \mid online \mid C_{\max}$, prior to \cite{Alb02},
Seiden~\cite{Sei03} modified the {\em non-barely} randomized algorithms of Bartal et al.~\cite{BFK95} and Seiden~\cite{Sei00},
both of which assign the current job to one of the two least loaded machines with certain probability,
to be barely randomized algorithms that are a {\em uniform} distribution of $k$ deterministic algorithms,
where $k$ is the number of schedules to be created inside the algorithm and was chosen big enough to guarantee the expected competitive ratios.
While succeeding for small $m$ (specifically, for $m \le 7$),
Albers~\cite{Alb02} observed that the analysis of the algorithms in \cite{BFK95,Sei00,Sei03} does not work for general large $m$,
and subsequently proposed the new barely randomized algorithm based on only two deterministic algorithms.

\subsection{Our results}
For the fully online and semi-online multiprocessor scheduling with testing to minimize the makespan,
we have seen that the only existing randomized algorithms are the optimal expected $\frac 43$-competitive algorithms
for the case of a single machine~\cite{DEM18,DEM20,AE20}.
Their expected competitive ratios are strictly smaller than the lower bound of $\varphi$ on the competitive ratio for any deterministic algorithm.
For the classic online multiprocessor scheduling,
only when $2 \le m \le 5$, the expected competitive ratios of the algorithms by Bartal et al.~\cite{BFK95} and Seiden~\cite{Sei00}
are strictly smaller than the corresponding lower bounds on the competitive ratio for any deterministic algorithm;
the expected competitive ratio $1.916$ of the algorithm by Albers~\cite{Alb02}
is also strictly smaller than the deterministic lower bound $1.919$, but {\em conditional} on using only the three offline makespan lower bounds.

In this paper, we always assume the presence of at least two machines in the fully online problem $P \mid online, t_j, 0 \le p_j \le u_j \mid C_{\max}$.
We propose a barely randomized algorithm that is an integration of {\em an arbitrary number} of component deterministic algorithms,
each is run with certain probability, and show that the expected competitive ratio is at most
$\left( \sqrt{(1 - \frac 1m)^2 \varphi^2 + 2 (1 - \frac 1m)} + 1 \right)$, where $m$ is the number of machines.
This expected competitive ratio is strictly increasing in $m$,
is strictly smaller than the competitive ratio $\varphi (2 - \frac 1m)$ of the best known deterministic algorithm by Albers and Eckl~\cite{AE21},
and approaches $(\sqrt{\varphi + 3} + 1) \approx 3.1490$ when $m$ tends to infinity.
To the best of our knowledge, barely randomized algorithms in the literature are mostly uniform distributions of its component deterministic algorithms,
while ours is the {\em first} non-uniform distribution of its {\em arbitrary many} component algorithms.

When there are only two machines, that is, for $P2 \mid online, t_j, 0 \le p_j \le u_j \mid C_{\max}$,
we show that employing two component deterministic algorithms in the randomized algorithm
leads to an expected competitive ratio of $\frac {3 \varphi + 3 \sqrt{13 - 7 \varphi}}4 \approx 2.1839$.
For this two-machine case, Albers and Eckl~\cite{AE21} proved a lower bound of $2.0953$ on the competitive ratio for any deterministic algorithm.
We improve it to $2.2117$,
thus showing that our randomized algorithm beats any deterministic algorithm in terms of the expected competitive ratio.
Such algorithmic results are rarely seen in the literature, except those we reviewed in the above.

It is worth noting that our randomized algorithm performs well for both large and small numbers of machines,
and better than the current best deterministic algorithms for all $m$'s in terms of the expected competitive ratios;
while most randomized multiprocessor scheduling algorithms in the literature work only for either a fixed and typically small number of machines
or for sufficiently large numbers of machines.
For example, the above mentioned Albers' randomized algorithm for the classic online multiprocessor scheduling requires large numbers of machines,
and the algorithms by Bartal et al.~\cite{BFK95} and Seiden~\cite{Sei00} perform well only when there are no more than seven machines.

On the inapproximability, we prove a lower bound of $1.6682$ on the expected competitive ratio of any randomized algorithm
for $P \mid online, t_j, 0 \le p_j \le u_j \mid C_{\max}$ at the presence of at least three machines,
and a lower bound of $1.6522$ on the expected competitive ratio for $P2 \mid online, t_j, 0 \le p_j \le u_j \mid C_{\max}$,
both using a slightly revised Yao's principle~\cite{Yao77}.
These two lower bounds are strictly better than the lower bound of $2 - \frac 1m$ by Albers and Eckl~\cite{AE21}
at the presence of three and two machines, respectively.
All our results are plotted in Figure~\ref{fig00} as well, highlighted with ``this paper''.

The rest of the paper is organized as follows.
In Section 2, we introduce some basic notations and definitions.
Next, in Section 3, we present the randomized algorithm for the fully online problem $P \mid t_j, 0 \le p_j \le u_j \mid C_{\max}$ and its performance analysis,
as well as the slightly modified randomized algorithm for the two-machine case and its performance analysis.
We prove in Section 4 the expected competitive ratio lower bounds of $1.6682$ and $1.6522$ for the problem
at the presence of at least three machines and only two machines, respectively,
and the lower bound of $2.2117$ on the competitive ratio for any deterministic algorithm for $P2 \mid online, t_j, 0 \le p_j \le u_j \mid C_{\max}$.
Lastly, we conclude the paper in Section 5.

\section{Preliminaries}
We study the fully online multiprocessor scheduling with testing to minimize the makespan,
denoted as $P \mid online, t_j, 0 \le p_j \le u_j \mid C_{\max}$, in which the jobs arrive in sequence.
Our goal is to design randomized algorithms with expected competitive ratios that are better than
the competitive ratios of the state-of-the-art deterministic algorithms,
or even further better than the best proven deterministic competitive ratio lower bounds.
The special case where the jobs all arrive at time zero is called semi-online and denoted as $P \mid t_j, 0 \le p_j \le u_j \mid C_{\max}$.
When applicable, we will prove lower bounds for the semi-online variant.

An instance $I$ of multiprocessor scheduling with testing consists of a job set $\mcJ = \{J_1, J_2, \ldots, J_n\}$,
each to be executed on one of a set of $m$ parallel identical machines $\mcM = \{M_1, M_2, \ldots, M_m\}$, 
where $n$ and $m \ge 2$ are part of the input.
Each job $J_j$ arrives with a known upper bound $u_j$ on its processing time $p_j$ and a testing operation of length $t_j$.
The $p_j$ remains unknown until the testing operation is executed.
That is, the scheduler can choose not to test the job $J_j$ but execute it for $u_j$ time,
or to test it for $t_j$ time followed by immediately execute it on the same machine for $p_j$ time.

A deterministic algorithm $A$ makes a binary decision on whether or not to test a job.
Let $p_j^A$ and $\rho_j$ denote the total time spent on the job $J_j$ in the algorithm $A$ and in the optimal offline schedule, respectively.
One sees that $p_j^A = t_j + p_j$ if the job is tested or otherwise $p_j^A = u_j$, and $\rho_j = \min\{u_j, t_j + p_j\}$.
Let $C_j$ denote the {\em completion time} of $J_j$ in the schedule generated by the algorithm $A$.
The objective of the problem is to minimize the makespan $C_{\max} = \max_j C_j$.
We use $C^A(I)$ (or $C^A$ when the instance $I$ is clear from the context) and $C^*(I)$ (or $C^*$, correspondingly)
to denote the makespan by the algorithm $A$ and of the optimal offline schedule for the instance $I$, respectively.
Switching to a more general randomized algorithm $A$, its expected makespan on the instance $I$ is denoted by $E[C^A(I)]$.

\begin{definition}
\label{def01}
{\rm (Competitive ratio)}
For a deterministic algorithm $A$, the competitive ratio is the worst-case ratio 
between the makespan of the schedule produced by the algorithm $A$ and of the optimal offline schedule,
that is, $\sup_{I} \{C^A(I) / C^*(I)\}$ where $I$ runs over all instances of the problem.

For a more general randomized algorithm $A$, its expected competitive ratio is $\sup_{I} \{E[C^A(I)] / C^*(I)\}$.
\end{definition}

Given that the processing time $p_j$ for the job $J_j$ could go to two extremes of $0$ and $u_j$,
a deterministic algorithm often sets up a threshold function on the ratio $r_j = \frac {u_j}{t_j}$ for binary testing decision making.
Intuitively, a randomized algorithm should tests the job $J_j$ with probability $f(r_j)$,
and such a probability function should be increasing in the ratio $r_j$ and $f(r_j) = 0$ for all $r_j \le 1$.
Indeed, D\"{u}rr et al.~\cite{DEM18,DEM20} and Albers and Eckl~\cite{AE20} used the following probability function 
\begin{equation}
\label{eq01}
f(r) = \left\{
\begin{array}{ll}
0, 											& \mbox{ if } r \le 1,\\
\frac{r(r - 1)}{r(r - 1) + 1},	& \mbox{ if } r > 1
\end{array}\right.
\end{equation}
in their optimal expected $\frac 43$-competitive randomized algorithms for the single machine problems
$P1 \mid online, t_j = 1, 0 \le p_j \le u_j \mid C_{\max}$ and $P1 \mid online, t_j, 0 \le p_j \le u_j \mid C_{\max}$, respectively.
Their success unfortunately cannot be easily extended to multiple machines,
since one key fact used in their performance analysis is that, for a single machine,
the expected makespan equals to the sum of the expected processing times of all the jobs.
We show below that, for multiple machines, if a randomized algorithm tests jobs using the probability function in Eq.~(\ref{eq01}), 
then its expected competitive ratio is unbounded when the number of machines tends to infinity.

\begin{lemma}
\label{lemma01}
For $P \mid online, t_j = 1, 0 \le p_j \le u_j \mid C_{\max}$,
if a randomized algorithm tests jobs using the probability function in Eq.~(\ref{eq01}),
then its expected competitive ratio is unbounded when the number of machines tends to infinity.
\end{lemma}
\begin{proof}
Consider an instance consisting of $m = k(k-1) + 1$ machines and $n = k(k-1) + 1$ jobs, where $k$ is a positive integer.
For each job $J_j$, $u_j = k$, $t_j = 1$ and $p_j = 0$.

Since $\rho_j = \min \{ u_j, t_j + p_j \} = 1$ for each job $J_j$,
in the optimal offline schedule, $J_j$ is tested and scheduled on the machine $M_j$, leading to a makespan of $1$.

The randomized algorithm tests each job with the same probability $\frac {k(k-1)}{k(k-1) + 1}$.
So the probability that at least one job is untested in the algorithm is 
\[
1 - \left(\frac {k(k-1)}{k(k-1) + 1}\right)^{k(k-1) + 1} = 1 - \left(1 - \frac {1}{k(k-1) + 1}\right)^{k(k-1) + 1}.
\]
Since the processing time of an untested job is $k$, the makespan is at least $k$.
It follows that the expected competitive ratio of the randomized algorithm on this instance is at least 
\[
k - k\left(1 - \frac {1}{k(k-1) + 1}\right)^{k(k-1) + 1}.
\]
When $m = k(k-1) + 1$ tends to infinity, the above becomes $(1 - \frac 1e) k$ and approaches infinity too.
\end{proof}

The proof of Lemma~\ref{lemma01} suggests to some extent that,
when the ratio $r_j$ of $J_j$ is large enough, one should test the job instead of leaving even only a tiny fraction of probability for untesting.
Indeed, later we will see that in our proposed randomized algorithm, there is a threshold on $r_j$'s for absolutely testing jobs.

On the other hand, at the presence of at least two machines,
we are no longer sure which jobs are assigned to each machine due to random job testing decisions if using a job testing probability function as in Eq.~(\ref{eq01}).
Therefore, we discard the randomized algorithm design idea in \cite{DEM18,DEM20,AE20},
but {\em extend} the idea of Albers~\cite{Alb02} to design {\em multiple} component deterministic algorithms $A_0, A_1, \ldots, A_{\ell}$ that complement each other,
where $\ell$ can be made arbitrarily large,
and then to run each $A_i$ with certain probability $\alpha_i$.
Let $p_j^{A_i}$ denote the total time spent on the job $J_j$ in $A_i$, for any $i = 0, 1, \ldots, \ell$.
The expected processing time of $J_j$ in the randomized algorithm $A$ is
\begin{equation}
\label{eq02}
E[p_j^A] = \sum_{i=0}^\ell \alpha_i p_j^{A_i}.
\end{equation}
Similarly, let $C^{A_i}$ be the makespan of the schedule produced by $A_i$, for any $i = 0, 1, \ldots, \ell$.
The expected makespan is
\begin{equation}
\label{eq03}
E[C^A] = \sum_{i=0}^\ell \alpha_i C^{A_i}.
\end{equation}

The following three lemmas hold for any component deterministic algorithms $A_0, A_1, \ldots, A_\ell$
run with the probability distribution $(\alpha_0, \alpha_1, \ldots, \alpha_\ell)$,
while assuming without loss of generality none of them tests any job $J_j$ with $r_j \le 1$.

\begin{lemma}
\label{lemma02}
For any $i = 0, 1, \ldots, \ell$, if $J_j$ with $r_j > 1$ is tested in $A_i$, then $p_j^{A_i} \le (1 + \frac 1{r_j}) \rho_j$.
\end{lemma}
\begin{proof}
Note that $p_j^{A_i} = t_j + p_j$.
If $\rho_j = t_j + p_j$, then $p_j^{A_i} = \rho_j$ and the lemma is proved.
Otherwise, $\rho_j = u_j$;
then from $p_j \le u_j$ and $u_j = r_j t_j$, we have 
$p_j^{A_i} = t_j + p_j \le (\frac 1{r_j} + 1) u_j = (\frac 1{r_j} + 1) \rho_j$ and the lemma is also proved.
\end{proof}

\begin{lemma}
\label{lemma03}
For any $i = 0, 1, \ldots, \ell$, if $J_j$ with $r_j > 1$ is untested in $A_i$, then $p_j^{A_i} \le r_j \rho_j$.
\end{lemma}
\begin{proof}
Note that $p_j^{A_i} = u_j$.
If $\rho_j = u_j$, then $p_j^{A_i} = \rho_j < r_j \rho_j$ and the lemma is proved.
Otherwise, $\rho_j = t_j + p_j$;
then from $p_j \ge 0$ and $u_j = r_j t_j$, we have 
$p_j^{A_i} = r_j t_j \le r_j (t_j + p_j) = r_j \rho_j$ and the lemma is also proved.
\end{proof}

\begin{lemma}
\label{lemma04}
If $J_j$ with $r_j > 1$ is tested in a subset ${\cal T}$ of the algorithms but untested in any of $\{A_0, A_1, \ldots, A_\ell\}\setminus {\cal T}$, then
\[
E[p_j^A] \le \max\left\{ \theta + (1-\theta) r_j, 1 + \frac {\theta}{r_j} \right\}\rho_j, \mbox{ where } \theta = \sum_{A_i \in {\cal T}} \alpha_i.
\]
\end{lemma}
\begin{proof}
By Eq.~(\ref{eq02}), $E[p_j^A] = \theta (t_j + p_j) + (1-\theta) u_j$.

If $\rho_j = t_j + p_j$, then $p_j^{A_i} = \rho_j$ for any $A_i \in {\cal T}$, and by Lemma~\ref{lemma03} $p_j^{A_i} \le r_j \rho_j$ for any $A_i \notin {\cal T}$;
therefore, $E[p_j^A] \le (\theta + (1-\theta) r_j)\rho_j$.
Otherwise, $\rho_j = u_j$;
then $p_j^{A_i} = \rho_j$ for any $A_i \notin {\cal T}$, and by Lemma~\ref{lemma02} $p_j^{A_i} \le (1 + \frac 1{r_j}) \rho_j$ for any $A_i \in {\cal T}$,
and therefore $E[p_j^A] \le \theta (1 + \frac 1{r_j}) \rho_j + (1 - \theta) \rho_j = (1 + \frac {\theta}{r_j})\rho_j$.
This proves the lemma.
\end{proof}

\begin{lemma}{\rm \cite{Alb02}}
\label{lemma05}
The following two lower bounds on the optimal offline makespan $C^*$ hold:
\begin{equation}
\label{eq04}
C^* \ge \max \left\{\frac1m \sum_{j=1}^n \rho_j, \max_{j=1}^n \rho_j\right\}.
\end{equation}
\end{lemma}

In the fully online problem $P \mid online, t_j, 0 \le p_j \le u_j \mid C_{\max}$,
the jobs arrive in sequence and it is assumed to be $\langle J_1, J_2, \dots, J_n\rangle$.
In the semi-online problem $P \mid t_j, 0 \le p_j \le u_j \mid C_{\max}$,
all the jobs arrive at time zero and an algorithm can take advantage of all the known $u_j$'s and $t_j$'s to process them in any order.

\section{Randomized algorithms} 
In this section, we first present a randomized algorithm for
the fully online multiprocessor scheduling with testing problem $P \mid online, t_j, 0 \le p_j \le u_j \mid C_{\max}$
at the presence of at least two machines,
and show that its expected competitive ratio is $(\sqrt{\varphi + 3} + 1) \approx 3.1490$.
When there are only two machines, we revise slightly the parameters in the algorithm leading to
an improved expected competitive ratio of $\frac {3 \varphi + 3 \sqrt{13-7\varphi}}4 \approx 2.1839$.

For the fully online multiprocessor scheduling with testing problem $P \mid online, t_j, 0 \le p_j \le u_j \mid C_{\max}$,
Albers and Eckl~\cite{AE21} proposed to test the job $J_j$ if $r_j \ge \varphi$ in the list scheduling rule~\cite{Gra66}
that assigns each job to the least loaded machine for processing (possible testing and then executing).
They showed that such a deterministic algorithm is tight $\varphi (2 - \frac 1m)$-competitive, where $m$ is the number of machines.
We employ it as our first component algorithm $A_0$,
and it is run with probability $\alpha_0 = \alpha(m, \ell)$ as in Eq.~(\ref{eq05}), where $\ell \ge 1$ is a fixed constant chosen by our randomized algorithm.
\begin{equation}
\label{eq05}
\alpha_0 = \alpha(m, \ell) = \sqrt{\frac {(1-\frac 1m) (\ell+1) \varphi^2}{(1-\frac 1m) (\ell+1) \varphi^2 + 2\ell}}, \ 
	\mbox{ and }
\alpha_i = \beta(m, \ell) = \frac {1-\alpha(m, \ell)}\ell, \mbox{ for } i = 1, 2, \ldots, \ell.
\end{equation}
Below we simplify $\alpha(m, \ell)$ and $\beta(m, \ell)$ as $\alpha$ and $\beta$, respectively.

One sees that when $r_j \le 1$, the job $J_j$ should not be tested, as this is the case in $A_0$.
On the other hand, when $r_j$ is large, testing does not waste too much time even if $p_j$ is very close to $u_j$, as this is also the case in $A_0$.
The ``risk'' of making wrong testing decision happens when $r_j$ is close to $\varphi$,
in that $p_j$ could be very close to $0$ if $J_j$ is untested while be very close to $u_j$ if tested.
Nevertheless, we have also seen from the proof of Lemma~\ref{lemma01} that
even a tiny fraction of probability for untesting jobs with large ratio $r_j$ could lead to an unbounded expected competitive ratio,
when the number $m$ of machines tends to infinity.
Therefore, in each of the other component algorithms $A_i$, $i = 1, 2, \ldots, \ell$,
we set up a threshold $y_i(m, \ell) > \varphi$ as in Eq.~(\ref{eq06}) such that a job $J_j$ with $r_j \ge y_i(m, \ell)$ is tested in $A_i$;
we then set up another threshold $x_i(m, \ell) = 1 + \frac 1{y_i(m, \ell)}$ such that a job $J_j$ with $r_j \le x_i(m, \ell)$ is untested in $A_i$.
\begin{equation}
\label{eq06}
y_i(m, \ell) = \frac {\varphi (\alpha + i \beta)}{\alpha}, \ 
x_i(m, \ell) = 1 + \frac 1{y_i(m, \ell)}, \mbox{ for } i = 0, 1, \ldots, \ell.
\end{equation}
Similarly, below we simplify $y_i(m, \ell)$ and $x_i(m, \ell)$ as $y_i$ and $x_i$, respectively.
When $r_j \in (x_i, y_i)$, we invert the testing decision as in $A_0$ to test $J_j$ if $r_j \le \varphi$ or untest $J_j$ if $r_j > \varphi$ in $A_i$.
Our randomized algorithm, denoted as {\sc GCL}, chooses to run $A_i$ with probability $\alpha_i$, for $i = 0, 1, \ldots, \ell$.

From Eqs.~(\ref{eq05}, \ref{eq06}), one sees that $\alpha, \beta > 0$, and
\begin{equation}
\label{eq07}
1 < x_\ell < \ldots < x_1 < x_0 = \varphi = y_0 < y_1 < \ldots < y_\ell = \frac {\varphi}{\alpha}.
\end{equation}
A high-level description of the algorithm {\sc GCL} is presented in Figure~\ref{GCL}.

\begin{figure}[htb]
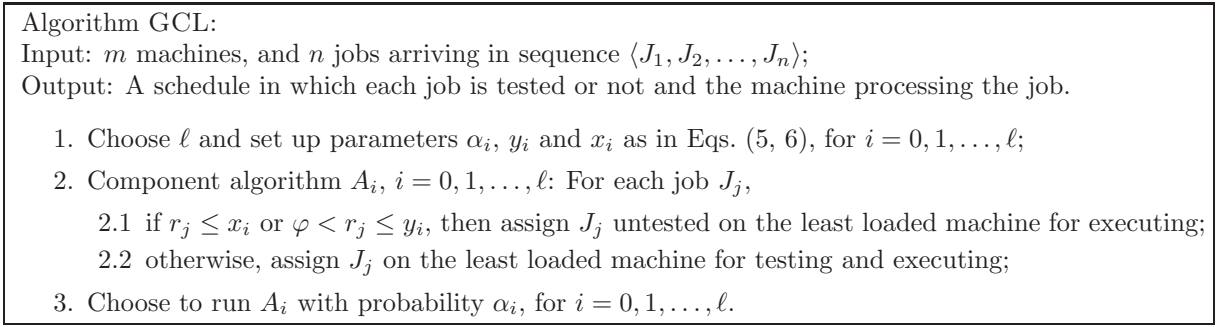

\begin{center}
\framebox{
\begin{minipage}{6.2in}
Algorithm {\sc GCL}:\\
Input:  $m$ machines, and $n$ jobs arriving in sequence $\langle J_1, J_2, \ldots, J_n\rangle$;\\
Output: A schedule in which each job is tested or not and the machine processing the job.
\begin{enumerate}
\parskip=0pt
\item
	Choose $\ell$ and set up parameters $\alpha_i$, $y_i$ and $x_i$ as in Eqs.~(\ref{eq05}, \ref{eq06}), for $i = 0, 1, \ldots, \ell$;
\item
	Component algorithm $A_i$, $i = 0, 1, \ldots, \ell$: For each job $J_j$,
	\begin{itemize}
	\parskip=0pt
	\item[2.1]
		if $r_j \le x_i$ or $\varphi < r_j \le y_i$, then assign $J_j$ untested on the least loaded machine for executing;
   \item[2.2]
     otherwise, assign $J_j$ on the least loaded machine for testing and executing;
	\end{itemize}
\item
   Choose to run $A_i$ with probability $\alpha_i$, for $i = 0, 1, \ldots, \ell$.
\end{enumerate}
\end{minipage}}
\end{center}
\caption{A high level description of the algorithm {\sc GCL}.\label{GCL}}
\end{figure}

The inequalities in the next lemma will be convenient in the performance analysis for the algorithm {\sc GCL}.

\begin{lemma}
\label{lemma08}
$\frac {1-\alpha}{x_\ell} < \frac {\alpha}{\varphi}$,
	\mbox{ and }
$(\ell - i) \beta (y_{i+1} - 1) < \frac {\alpha}{\varphi}$ for every $i = 0, 1, \ldots, \ell-1$.
\end{lemma}
\begin{proof}
By its definition in Eq.~(\ref{eq05}), $\alpha$ is strictly increasing in $m$.
Therefore,
\begin{equation}
\label{eq08}
\sqrt{\frac {(\ell + 1) \varphi^2}{(\ell + 1) \varphi^2 + 2\ell}} > \alpha = \alpha(m, \ell)
	\ge \alpha(2, \ell) = \sqrt{\frac {(\ell + 1) \varphi^2}{(\ell + 1) \varphi^2 + 4\ell}}
	\ge \sqrt{\frac {\varphi^2}{\varphi^2 + 4}} > \varphi - 1.
\end{equation}
Using $y_\ell = \frac {\varphi}{\alpha}$ and $x_\ell = 1 + \frac {\alpha}{\varphi}$, we have $x_\ell > 1 + \frac {\varphi-1}{\varphi} = 3 - \varphi$.
It follows that $(\varphi + x_\ell) \alpha > 3 (\varphi - 1) > \varphi$, and thus
\[
\frac {\alpha}{\varphi} - \frac {1-\alpha}{x_\ell} = \frac {(\varphi + x_\ell) \alpha - \varphi}{\varphi x_\ell} > 0,
\]
which is the first inequality.

Define $g(i) = (\ell - i) \beta (y_{i+1} - 1)$ for $i = 0, 1, \ldots, \ell-1$, which is a concave quadratic function in $i$
and its axis of symmetry is $i = i_0 = \frac 12 (\ell - 1 - \frac {\alpha}{(\varphi + 1) \beta})$.
Note that when $\ell = 1$, Eq.~(\ref{eq08}) implies $\sqrt{\frac {\varphi^2}{\varphi^2 + 1}} > \alpha \ge \sqrt{\frac {\varphi^2}{\varphi^2 + 2}}$;
thus by Eq.~(\ref{eq06}), we have
$g(0) = \beta (\frac {\varphi}{\alpha} - 1) = (1 - \alpha) (\frac {\varphi}{\alpha} - 1) < \frac {\alpha}{\varphi}$ and the lemma is proved.
When $\ell = 2$, Eq.~(\ref{eq08}) implies $\sqrt{\frac {3\varphi^2}{3\varphi^2 + 4}} > \alpha \ge \sqrt{\frac {3\varphi^2}{3\varphi^2 + 8}}$;
thus by Eq.~(\ref{eq06}), we have
$g(i) \le 2 \beta (\frac {\varphi}{\alpha} - 1) = (1 - \alpha) (\frac {\varphi}{\alpha} - 1) < \frac {\alpha}{\varphi}$ for $i = 0, 1$ and the lemma is proved.

Below we assume that $\ell \ge 3$.
Note that when $i \le (2 - \varphi) \ell - 1$, Eqs.~(\ref{eq05}, \ref{eq06}, \ref{eq08}) together imply
\[
y_{i+1} \le \frac {\varphi (\alpha + (2 - \varphi)\ell \beta)}{\alpha} = \frac {\alpha + \varphi-1}{\alpha} < 2
\]
and thus
\[
g(i) \le (\ell - i) \beta \le 1 - \alpha < 2 - \varphi = \frac {\varphi-1}{\varphi} < \frac {\alpha}{\varphi}.
\]
Using $\beta = \frac 1{\ell} (1 - \alpha) \le \frac 13 (1 - \alpha)$ and Eq.~(\ref{eq08}), one sees that
\[
(2 - \varphi) \ell - 1 - i_0 = \frac {\varphi \alpha - (\varphi - 1) - (\varphi + 1) \beta}{2 (\varphi + 1) \beta}
	> \frac {5 - 3 \varphi}{6 (\varphi + 1) \beta} > 0,
\]
that is, $(2 - \varphi) \ell - 1$ is to the right of the axis of symmetry.
Therefore, for every $i = 0, 1, \ldots, \ell-1$, $g(i) < \frac {\alpha}{\varphi}$.
This completes the proof of the lemma.
\end{proof}

\begin{lemma}
\label{lemma09}
In the algorithm $A_i$, $p_j^{A_i}\le y_i \rho_j$ for every job $J_j$, for any $i = 0, 1, \ldots, \ell$.
\end{lemma}
\begin{proof}
In the algorithm $A_0$, a job $J_j$ is untested if and only if $r_j \le \varphi$.
By Lemmas~\ref{lemma02} and \ref{lemma03}, $p_j^{A_0} \le \max\{ 1 + \frac 1{\varphi}, \varphi\} \rho_j = y_0 \rho_j$.

Consider the algorithm $A_i$, for any $i = 1, 2, \ldots, \ell$.
If $r_j \le x_i$ or $\varphi < r_j \le y_i$, then $J_j$ is untested in $A_i$ and, by Lemma~\ref{lemma03}, $p_j^{A_i} \le y_i \rho_j$.
Otherwise, $x_i < r_j \le \varphi$ or $r_j > y_i$, $J_j$ is tested and, by Lemma~\ref{lemma02}, $p_j^{A_i} \le (1 + \frac 1{x_i}) \rho_j$.
In conclusion, for each job $J_j$, we have 
\[
p_j^{A_i} \le \max \left\{ y_i, 1 + \frac 1{x_i} \right\} \rho_j = \max \left\{ y_i, 1 + \frac {y_i}{y_i + 1} \right\} \rho_j = y_i \rho_j,
\]
since $y_i \ge \varphi$.
This proves the lemma.
\end{proof}

\begin{lemma}
\label{lemma10}
In the algorithm {\sc GCL}, $E[p_j^A] \le x_\ell \rho_j$ for every job $J_j$.
\end{lemma}
\begin{proof}
If $r_j > y_\ell$, then $J_j$ is tested in every algorithm $A_i$, for $i = 0, 1, \ldots, \ell$.
So, by Lemma~\ref{lemma04}, $E[p_j^A] \le (1 + \frac 1{y_\ell}) \rho_j = x_\ell \rho_j$.
If $r_j \le x_\ell$, then $J_j$ is untested in every algorithm $A_i$, $i = 0, 1, \ldots, \ell$.
So, by Lemma~\ref{lemma04}, $E[p_j^A] \le x_\ell \rho_j$.

If $y_h < r_j \le y_{h+1}$, for some $h = 0, 1, \ldots, \ell-1$,
then $J_j$ is tested in $A_0, A_1, \ldots, A_h$ but untested in $A_{h+1}, A_{h+2}, \ldots, A_\ell$.
Lemma~\ref{lemma04} and then Eq.~(\ref{eq06}) and Lemma~\ref{lemma08} together imply that
\begin{eqnarray*}
E[p_j^A] &\le & \max\left\{ \alpha + h \beta + (\ell - h) \beta  y_{h+1}, 1 + \frac {\alpha + h \beta}{y_h} \right\} \rho_j\\
	&=& \max\left\{ 1 + (\ell - h) \beta (y_{h+1} - 1), 1 + \frac {\alpha}{\varphi} \right\}\rho_j\\
	&\le & (1 + \frac {\alpha}{\varphi}) \rho_j = x_\ell \rho_j.
\end{eqnarray*}

Lastly, if $x_{h+1} < r_j \le x_h$, for some $h = 0, 1, \ldots, \ell-1$,
then $J_j$ is untested in $A_0, A_1, \ldots, A_h$ but tested in $A_{h+1}, A_{h+2}, \ldots, A_\ell$.
Lemma~\ref{lemma04} and then Eq.~(\ref{eq06}), Lemma~\ref{lemma08} and Eq.~(\ref{eq08}) together imply that
\begin{eqnarray*}
E[p_j^A] &\le & \max\left\{ (\ell - h) \beta + (1 - (\ell - h) \beta) x_h, 1 + \frac {(\ell - h) \beta}{x_{h+1}} \right\} \rho_j\\
	&=& \max \left\{ 1 + \frac {\alpha}{\varphi}, 1 + \frac {(\ell - h) \beta}{x_{h+1}} \right\} \rho_j\\
	&\le& \max \left\{ 1 + \frac {\alpha}{\varphi}, 1 + \frac {1 - \alpha}{x_\ell} \right\} \rho_j\\
	&\le & (1 + \frac {\alpha}{\varphi})\rho_j = x_\ell \rho_j.
\end{eqnarray*}
This proves the lemma.
\end{proof}

\begin{theorem}
\label{thm04}
The expected competitive ratio of {\sc GCL} for the problem $P \mid online, t_j, 0 \le p_j \le u_j \mid C_{\max}$
is at most $\sqrt{(1 - \frac 1m)^2 (1 + \frac 1\ell)^2 \varphi^2 + 2(1 - \frac 1m)(1 + \frac 1\ell)} + 1 - (1 - \frac 1m) \frac \varphi\ell$
by setting $\alpha(m, \ell), \beta(m, \ell)$ as in Eq.~(\ref{eq05}) and $y_i(m, \ell), x_i(m, \ell)$ for every $i = 0, 1, \ldots, \ell$ as in Eq.~(\ref{eq06}),
where $m \ge 2$ is the number of machines and $\ell \ge 1$ is a fixed constant.
\end{theorem}
\begin{proof}
Suppose the job $J_{n_i}$ determines the makespan of the schedule produced by the component algorithm $A_i$,
for each $i = 0, 1, \ldots, \ell$, respectively.
Note that $J_{n_i}$ is assigned to the least loaded machine by the algorithm $A_i$.
Therefore, the makespan
\[
C^{A_i} \le \frac 1m \sum_{j=1}^{n_i - 1} p_j^{A_i} + p_{n_i}^{A_i} = \frac 1m \sum_{j=1}^{n_i} p_j^{A_i} + (1 - \frac 1m) p_{n_i}^{A_i}
	\le \frac 1m \sum_{j=1}^n p_j^{A_i} + (1 - \frac 1m) p_{n_i}^{A_i}.
\]
Using Eqs.~(\ref{eq03}),~(\ref{eq07}), Lemmas~\ref{lemma09}, \ref{lemma10} and the lower bounds in Lemma \ref{lemma05},
the expected makespan is at most
\begin{eqnarray*}
E[C^A] & = & \alpha C^{A_0} + \sum_{i=1}^\ell \beta C^{A_i} \\
	& \le & \frac 1m \sum_{j = 1}^n E[p_j^A] + \alpha (1 - \frac 1m) p^{A_0}_{n_0} + \sum_{i=1}^\ell \beta (1 - \frac 1m) p^{A_i}_{n_i} \\
	& \le & \frac {x_\ell}m \sum_{j = 1}^n \rho_j + \alpha (1 - \frac 1m) \varphi \rho_{n_0} + \sum_{i=1}^\ell \beta (1 - \frac 1m) y_i \rho_{n_i} \\
	& \le & \left( x_\ell + (1 - \frac 1m) \left(\alpha \varphi + \beta \sum_{i=1}^\ell y_i\right) \right) C^* \\
	& = & \left( 1 + \frac {\alpha}{\varphi} + (1 - \frac 1m) \left( \alpha \varphi + \ell \beta \varphi + \frac {\ell (\ell + 1) \beta^2 \varphi}{2 \alpha}\right) \right) C^* \\
	& = & \left( \sqrt{(1-\frac 1m)^2 (1 + \frac 1\ell)^2 \varphi^2 + 2 (1-\frac 1m) (1 + \frac 1\ell)} + 1 - (1 - \frac 1m) \frac {\varphi}{\ell} \right) C^*,
\end{eqnarray*}
and thus the theorem is proved.
\end{proof}

By choosing a sufficiently large $\ell$, the expected competitive ratio of the algorithm {\sc GCL} approaches $\sqrt{(1 - \frac 1m)^2 \varphi^2 + 2 (1 - \frac 1m)} + 1$,
which is increasing in $m$ and approaches $\sqrt{\varphi + 3} + 1 \approx 3.1490$ when $m$ tends to infinity.
Correspondingly, one can check that $\alpha$, $y_\ell$ and $x_\ell$ approach  
$\sqrt{\frac {\varphi + 1}{\varphi + 3}} \approx 0.7529$, $\sqrt{\varphi + 3} \approx 2.1490$ and $1 + \frac 1{\sqrt{\varphi + 3}} \approx 1.4653$, respectively.

In the algorithm {\sc GCL}, we can determine the testing probability $f(r)$ of a job as a function in its ratio $r$.
For example, $f(r) = 1$ if $r \ge y_\ell$ and $f(r) = 0$ if $r \le x_\ell$.
With a fixed constant $\ell$, one sees that such a probability function $f(r)$ is staircase and increasing.
The limit of this probability function when $\ell$ tends to infinity, still denoted as $f(r)$, is interesting.
First, note that $\alpha(m, \ell)$ approaches $\alpha(m) = \sqrt{ \frac {(1-\frac 1m) \varphi^2}{(1-\frac 1m) \varphi^2 + 2} }$,
$y_\ell(m, \ell)$ and $x_\ell(m, \ell)$ approach $y(m) = \frac \varphi{\alpha(m)}$ and $x(m) = 1 + \frac 1{y(m)}$, respectively.
Then, by Eq.(\ref{eq06}), we have
\begin{equation*}
f(r) = \left\{
\begin{array}{ll}
0, 														& \mbox{ when } r \le x(m),\\
1 - \frac{\alpha(m)}{\varphi (r-1)},	& \mbox{ when } x(m) < r \le \varphi, \\
\frac{\alpha(m) r}\varphi,					& \mbox{ when } \varphi < r \le y(m), \\
1, 														& \mbox{ when } r > y(m),
\end{array}\right.
\end{equation*}
which is increasing, and continuous everywhere except at $\varphi$.

We remark that our {\sc GCL} algorithm for $P \mid online, t_j, 0 \le p_j \le u_j \mid C_{\max}$
has its expected competitive ratio beating the best known deterministic competitive ratio of $\varphi (2 - \frac 1m)$.

When there are only two machines, i.e., $m = 2$,
by choosing a sufficiently large $\ell$, the algorithm {\sc GCL} is already expected $\frac 12 \sqrt{\varphi + 5} + 1 \approx 2.2863$-competitive
for the problem $P2 \mid online, t_j, 0 \le p_j \le u_j \mid C_{\max}$.
Note that inside the algorithm, the probability of choosing the algorithm $A_0$ approaches $\alpha(2) = \frac \varphi{\sqrt{\varphi + 5}} \approx 0.6290$.
We next show that, by choosing another probability for running $A_0$,
the revised {\sc GCL} algorithm based on two component algorithms is expected $\frac {3\varphi + 3\sqrt{13 - 7\varphi}}4 \approx 2.1839$-competitive.

\begin{theorem}
\label{thm05}
The expected competitive ratio of the revised {\sc GCL} algorithm for the problem $P2 \mid online, t_j, 0 \le p_j \le u_j \mid C_{\max}$
is at most $\frac {3\varphi + 3\sqrt{13 - 7\varphi}}4$ by using two component algorithms and by setting
\begin{equation}
\label{eq09}
\alpha = \varphi - 1 \approx 0.6180,
x_1 = \frac {\varphi + \sqrt{13 - 7\varphi}}2 \approx 1.4559, \mbox{ and } 
y_1 = \frac 1{x_1-1} \approx 2.1935.
\end{equation}
Furthermore, the expected competitive ratio is tight.
\end{theorem}
\begin{proof}
The three parameters in Eq.~(\ref{eq09}) satisfy the algorithm design: $1 < x_1 < \varphi < y_1$ and $\alpha < 1$.

Furthermore, they imply that $1 - \alpha + \alpha \varphi = 3 - \varphi < x_1$ and
$1 + \frac {\alpha}{\varphi} = 2 - \frac 1{\varphi} < x_1$.
Also noticing that $x_1$ is a root to the quadratic equation $x^2 - \varphi x + 2\varphi - 3 = 0$,
we have $x_1^2 = \varphi x_1 + 3 - 2\varphi$.
It follows that
\[
y_1 - 1 - \frac 1{x_1} = \frac {1 + x_1 - x_1^2}{x_1 (x_1-1)} = \frac {(\varphi - 1)(2 - x_1)}{x_1 (x_1-1)} > 0,
\]
\[
x_1 - 1 - \frac {1 - \alpha}{x_1} = \frac {x_1^2 - x_1 - (2 - \varphi)}{x_1} = \frac {(\varphi - 1)(x_1 - 1)}{x_1} > 0,
\]
and 
\[
x_1 - \alpha - (1 - \alpha) y_1 = \frac {x_1^2 - \varphi x_1 + 2\varphi - 3}{x_1 - 1} = 0.
\]

For each job $J_j$, if $r_j \le x_1$, then it is not tested in any of $A_0$ and $A_1$ and thus $p_j^{A_i} = u_j$, for $i = 0, 1$.
When $r_j \le 1$, $\rho_j = u_j$ and thus $p_j^{A_i} = \rho_j$;
when $1 < r_j \le x_1$, by Lemma~\ref{lemma03} we have $p_j^{A_i} \le x_1 \rho_j$.
That is, if $r_j \le x_1$, then we always have
\begin{equation}
\label{eq10}
p_j^{A_i} \le x_1 \rho_j, \mbox{ for } i = 0, 1.
\end{equation}

If $r_j > y_1$, then $J_j$ is tested in both $A_0$ and $A_1$, and thus by Lemma~\ref{lemma02} and Eq.~(\ref{eq09}) we have
\begin{equation}
\label{eq11}
p_j^{A_i} \le (1 + \frac 1{y_1})\rho_j = x_1 \rho_j, \mbox{ for } i = 0, 1.
\end{equation}

If $x_1 < r_j \le \varphi$, then $J_j$ is untested in $A_0$ but tested in $A_1$.
By Lemmas~\ref{lemma02}--\ref{lemma04}, we have
\begin{equation}
\label{eq12}
p_j^{A_0} \le \varphi \rho_j \mbox{, }
p_j^{A_1} \le (1+\frac 1{x_1}) \rho_j \mbox{, and }
E[p_j^A] \le \max \left\{ 1-\alpha + \alpha \varphi, 1 + \frac {1-\alpha}{x_1} \right\} \rho_j.
\end{equation}

Lastly, if $\varphi < r_j \le y_1$, then $J_j$ is tested in $A_0$ but untested in $A_1$.
By Lemmas~\ref{lemma02}--\ref{lemma04}, we have
\begin{equation}
\label{eq13}
p_j^{A_0} \le \varphi \rho_j \mbox{, }
p_j^{A_1} \le y_1 \rho_j \mbox{, and }
E[p_j^A] \le \max \left\{ \alpha + (1-\alpha) y_1, 1 + \frac {\alpha}\varphi \right\} \rho_j.
\end{equation}

The processing time of $J_j$ in the two algorithms in all the four cases, in Eqs.~(\ref{eq10}--\ref{eq13}), can be merged as
\begin{equation}
\label{eq14}
p_j^{A_0} \le \varphi \rho_j \mbox{ and }
p_j^{A_1} \le \max\left\{ 1+\frac 1{x_1}, y_1 \right\} \rho_j = y_1 \rho_j;
\end{equation}
and the expected processing time of $J_j$ in the revised {\sc GCL} algorithm in all the four cases can be merged as
\begin{equation}
\label{eq15}
E[p_j^A] \le \max\left\{ x_1, 1 + \frac {1-\alpha}{x_1}, \alpha + (1-\alpha) y_1 \right\} \rho_j = x_1 \rho_j.
\end{equation}

Suppose the job $J_{n_i}$ determines the makespan of the schedule produced by the algorithm $A_i$, for $i = 0, 1$, respectively.
Note that $J_{n_i}$ is assigned to the least loaded machine by the algorithm.
Therefore,
\begin{equation}
\label{eq16}
C^{A_i} \le \frac 12 \sum_{j=1}^{n_i - 1} p_j^{A_i} + p_{n_i}^{A_i} = \frac 12 \sum_{j=1}^{n_i} p_j^{A_i} + \frac 12 p_{n_i}^{A_i}.
\end{equation}
We distinguish the following three cases.

\noindent{\bf Case 1:} $n_1 = n_2$.
Using Eqs.~(\ref{eq14}, \ref{eq15}, \ref{eq16}) and the lower bounds in Lemma~\ref{lemma05}, we have
\begin{eqnarray*}
E[C^A] & = & \alpha C^{A_0} + (1-\alpha) C^{A_1} \\
	& \le & \frac 12 \sum_{j = 1}^{n_1} (\alpha p_j^{A_0} + (1-\alpha)p_j^{A_1}) + \frac 12 (\alpha p_{n_1}^{A_0} + (1-\alpha)p_{n_1}^{A_1}) \\
	& = &  \frac 12 \sum_{j = 1}^{n_1} E[p_j^A] + \frac 12 E[p_{n_1}^A] \\
	& \le & \frac {x_1}2 \sum_{j = 1}^{n_1} \rho_j + \frac {x_1}2 \rho_{n_1} \\
	& \le & \frac {3x_1}2 C^*.
\end{eqnarray*}

\noindent{\bf Case 2:} $n_1 > n_2$.
Eq.~(\ref{eq16}) becomes
\[
C^{A_0} \le \frac 12 \sum_{j=1}^{n_1} p_j^{A_0} + \frac 12 p_{n_1}^{A_0}, \mbox{ and }
C^{A_1} \le \frac 12 \sum_{j=1}^{n_1} p_j^{A_1} + \frac 12 p_{n_2}^{A_1} - \frac 12 p_{n_1}^{A_1}.
\]
Then using $p_{n_1}^{A_1} \ge \rho_{n_1}$, Eqs.~(\ref{eq09}, \ref{eq14}, \ref{eq15}) and the lower bounds in Lemma~\ref{lemma05}, we have
\begin{eqnarray*}
E[C^A] & = & \alpha C^{A_0} + (1-\alpha) C^{A_1} \\
	& \le & \frac 12 \sum_{j = 1}^{n_1} E[p_j^A] + \frac {\alpha}2  p_{n_1}^{A_0} + \frac {1-\alpha}2 p_{n_2}^{A_1} - \frac {1-\alpha}2 p_{n_1}^{A_1} \\
	& \le & \frac {x_1}2 \sum_{j = 1}^{n_1} \rho_j + \frac {\alpha \varphi}2 \rho_{n_1} + \frac {1-\alpha}2 y_1 \rho_{n_2} - \frac {1-\alpha}2 \rho_{n_1} \\
	& = & \frac {x_1}2 \sum_{j = 1}^{n_1} \rho_j + \frac {\varphi-1}2 \rho_{n_1} + \frac {x_1-\varphi+1}2 \rho_{n_2} \\
	& \le & \frac {3x_1}2 C^*.
\end{eqnarray*}

\noindent{\bf Case 3:} $n_1 < n_2$. 
Eq.~(\ref{eq16}) becomes
\[
C^{A_0} \le \frac 12 \sum_{j=1}^{n_2} p_j^{A_0} + \frac 12 p_{n_1}^{A_0} - \frac 12 p_{n_2}^{A_0}, \mbox{ and }
C^{A_1} \le \frac 12 \sum_{j=1}^{n_2} p_j^{A_1} + \frac 12 p_{n_2}^{A_1}.
\]
Then using $p_{n_2}^{A_0} \ge \rho_{n_2}$, Eqs.~(\ref{eq09}, \ref{eq14}, \ref{eq15}) and the lower bounds in Lemma~\ref{lemma05}, we have
\begin{eqnarray*}
E[C^A] & = & \alpha C^{A_0} + (1-\alpha) C^{A_1} \\
	& \le & \frac 12 \sum_{j = 1}^{n_2} E[p_j^A] + \frac {\alpha}2  p_{n_1}^{A_0} + \frac {1-\alpha}2 p_{n_2}^{A_1} - \frac {\alpha}2 p_{n_2}^{A_0} \\
	& \le & \frac {x_1}2 \sum_{j = 1}^{n_2} \rho_j + \frac {\alpha \varphi}2 \rho_{n_1} + \frac {1-\alpha}2 y_1 \rho_{n_2} - \frac {\alpha}2 \rho_{n_2} \\
	& = & \frac {x_1}2 \sum_{j = 1}^{n_2} \rho_j + \frac 12 \rho_{n_1} + \frac {x_1-2\varphi+2}2 \rho_{n_2} \\
	& \le & \left( \frac {3x_1}2 + \frac {3-2\varphi}2 \right)C^* \\
	& \le & \frac {3x_1}2 C^*.
\end{eqnarray*}

In all the above three cases, we have $E[C^A] \le \frac {3x_1}2 C^* = \frac {3\varphi+3\sqrt{13-7\varphi}}4 C^*$.
This proves that the expected competitive ratio of the revised {\sc GCL} algorithm is at most $\frac {3\varphi+3\sqrt{13-7\varphi}}4 \approx 2.1839$.

We next give a three-job instance of the two-machine problem $P2 \mid online, t_j, 0 \le p_j \le u_j \mid C_{\max}$ in Instance~\ref{ex04}
to show the expected competitive ratio of the revised {\sc GCL} algorithm is tight.

\begin{instance} 
\label{ex04}
In this three-job instance of $P2 \mid online, t_j, 0 \le p_j \le u_j \mid C_{\max}$,
the job order is $\langle J_1, J_2, J_3\rangle$, with
\[
\left\{
\begin{array}{lll}
u_1 = x_1, 		&t_1 = 1, 	&p_1 = 0,\\ 
u_2 = x_1, 		&t_2 = 1, 	&p_2 = 0,\\ 
u_3 = 2 x_1, 	&t_3 = 2, 	&p_3 = 0.
\end{array}
\right.
\]
\end{instance}

One sees that for this instance, $\rho_1 = \rho_2 = 1$ and $\rho_3 = 2$.
Therefore, the optimal offline makespan is $C^* = 2$ by testing all the jobs and assigning $J_1, J_2$ on one machine and $J_3$ on the other machine.
On the other hand, $r_j = x_1 < \varphi$ for any $j = 1, 2, 3$ and so they are all untested in either algorithm of $A_0$ and $A_1$.
It follows that $C^{A_0} = C^{A_1} = 3 x_1$ and thus $E[C^A] = 3x_1$ too.
The expected competitive ratio of {\sc GCL} on this instance is $\frac {3 x_1}2$.
This finishes the proof of the theorem.
\end{proof}

\section{Lower bounds on competitive ratios} 
In this section, we show lower bounds of $1.6682$ and $1.6522$ on the expected competitive ratio of any randomized algorithm
for the fully online problem $P \mid t_j, 0 \le p_j \le u_j \mid C_{\max}$ at the presence of at least three machines and only two machines, respectively.
One thus sees that the expected competitive ratios of the {\sc GCL} algorithm are far away from
these two lower bounds and $2 - \frac 1m$~\cite{AE21} on the expected competitive ratio at the presence of two, three and $m \ge 4$ machines, respectively.
These gaps certainly welcome future research.
We also prove an improved lower bound of $2.2117$ on the competitive ratio for any deterministic algorithm for $P2 \mid t_j, 0 \le p_j \le u_j \mid C_{\max}$,
implying that the {\sc GCL} algorithm is better than any deterministic algorithm in terms of expected competitive ratio.

\subsection{On expected competitive ratios}
Yao's principle~\cite{Yao77} is an efficient tool in proving the lower bounds for randomized algorithms.
Below we prove a theorem which is a slight modification of Yao's principle.
Let $\mcI$ and $\mcA$ be the deterministic instance space and the deterministic algorithm space, respectively, for an optimization problem.
For distinction, let $I$ and $I_r$ denote a deterministic instance and a randomized instance, respectively,
and let $A$ and $A_r$ denote a deterministic algorithm and a randomized algorithm, respectively.
Note that in a randomized instance $I_r$ of the multiprocessor scheduling with testing problem,
each job $J_j$ is given known $u_j$ and $t_j$ values and, after the testing is done, a probability distribution of $p_j$ value.

\begin{theorem}
\label{thm01}
For any randomized algorithm $A_r$ and any randomized instance $I_r$, the expected competitive ratio of $A_r$ satisfies
\[
\sup_{I \in \mcI} \frac{E[C^{A_r}(I)]}{C^*(I)} \ge \inf_{A \in \mcA} \frac{E[C^A(I_r)]}{E[C^*(I_r)]},
\]
where $E[C^A(I_r)]$ and $E[C^*(I_r)]$ are the expected objective value for a deterministic algorithm $A$ and
the expected optimal objective value on the instance $I_r$, respectively.
\end{theorem}
\begin{proof}
For any deterministic algorithm $A$, we have
\[
E[C^A(I_r)] = \sum_{I \in \mcI} Pr[I_r = I] C^A(I),
\]
where $Pr[I_r = I]$ is the probability that the randomized instance $I_r$ happens to be the deterministic instance $I$.
Note that $\sum_{A \in \mcA} Pr[A_r = A] = 1$ and $E[C^{A_r}(I)] = \sum_{A \in \mcA} Pr[A_r = A] C^A(I)$.
It follows that
\begin{eqnarray*}
\inf_{A \in \mcA} E[C^A(I_r)] & = & \inf_{A \in \mcA} \sum_{I \in \mcI} Pr[I_r = I] C^A(I) \\
	& \le & \sum_{A \in \mcA}Pr[A_r = A] \sum_{I \in \mcI} Pr[I_r = I] C^A(I)\\
	& = & \sum_{I \in \mcI}Pr[I_r = I] \sum_{A \in \mcA} Pr[A_r = A] C^A(I)\\
	& = & \sum_{I \in \mcI}Pr[I_r = I] E[C^{A_r}(I)].
\end{eqnarray*}
Since $E[C^*(I_r)] = \sum_{I \in \mcI} Pr[I_r = I] C^*(I)$, at the end we obtain
\[
\inf_{A \in \mcA} \frac{E[C^A(I_r)]}{E[C^*(I_r)]} 
	\le \frac{\sum_{I \in \mcI}Pr[I_r = I] E[C^{A_r}(I)]}{\sum_{I \in \mcI} Pr[I_r = I] C^*(I)} 
	\le \sup_{I \in \mcI} \frac{E[C^{A_r}(I)]}{C^*(I)},
\]
which proves the theorem.
\end{proof}

For the semi-online single machine uniform testing problem $P1 \mid t_j = 1, 0 \le p_j \le u_j \mid C_{\max}$,
the following randomized instance shows the tight lower bound of $\frac 43$ on expected competitive ratios~\cite{DEM18,DEM20,AE20}.

\begin{instance} 
\label{ex01}
In this randomized instance of $P1 \mid t_j = 1, 0 \le p_j \le u_j \mid C_{\max}$,
there is only a single job $J_1$ with $u_1 = 2$ and $t_1 = 1$, and $p_1 = 0$ or $2$ each of probability $0.5$.

In any deterministic algorithm $A$, $J_1$ is either tested or untested.
If $J_1$ is untested, then $p_1^A = u_1 = 2$;
otherwise $J_1$ is tested and thus $p_1^A = t_1 + p_1$ which is $1$ or $3$ each of probability $0.5$.
It follows that the expected total execution time is $E[p_1^A] = 2$, disregarding $J_1$ is tested or not.

On the other hand, the optimal offline execution time is $\rho_1 = 1$ if $p_1 = 0$, or $2$ if $p_1 = 2$;
that is, the expected optimal execution time is $E[\rho_1] = 1.5$.
By Theorem~\ref{thm01}, $\frac {E[p_1^A]}{E[\rho_1]} = \frac 43$ is a lower bound on expected competitive ratios for
the semi-online single machine uniform testing problem $P1 \mid t_j = 1, 0 \le p_j \le u_j \mid C_{\max}$.
\end{instance}

When there are multiple machines,
the following deterministic instance shows a lower bound of $2 - \frac 1m$ on expected competitive ratios~\cite{AE21}.

\begin{instance} 
\label{ex02}
In this deterministic instance of $P \mid t_j = 1, 0 \le p_j \le u_j \mid C_{\max}$,
every job has $u_i = M$ and $t_i = 1$, where $M$ is sufficiently large,
there are $m(m-1)$ small jobs with $p_i = 0$,
and additionally a large job with its $p_i = m-1$.

The optimal offline schedule tests all the jobs, giving rise to $C^* = m$.

On the other hand, the sufficiently large $M$ forces any deterministic algorithm to also test all the jobs.
Noting that all these jobs are indistinguishable at the arrival.
The adversary decides the processing time of a particular job to be $m-1$,
which is the first job assigned by the algorithm to a machine that has a load of $m-1$.
This way, the makespan of the generated schedule is $C_{\max} \ge 2m-1$.

For the randomized instance consisting of only this deterministic instance,
Theorem~\ref{thm01} implies that $2 - \frac 1m$ is a lower bound on expected competitive ratios for
the semi-online multiprocessor uniform testing problem $P \mid t_j = 1, 0 \le p_j \le u_j \mid C_{\max}$.
\end{instance}

The above two lower bounds on expected competitive ratios for the most special semi-online uniform testing case
hold surely for the fully online general testing case too.
In the next two theorems, we prove lower bounds specifically for the fully online general testing problem $P \mid online, t_j, 0 \le p_j \le u_j \mid C_{\max}$,
using randomized instances constructed out of the following deterministic instance.
These two lower bounds are better than $2 - \frac 1m$ when $m = 2, 3$.

\begin{instance} 
\label{ex03}
In this deterministic instance $I$ of $P \mid online, t_j, 0 \le p_j \le u_j \mid C_{\max}$,
there are $m$ machines and $m+1$ jobs arriving in sequence $\langle J_1, J_2, \cdots, J_{m+1}\rangle$, with
\[
u_j = 2, t_j = 1, \mbox{ for any } j = 1, 2, \cdots, m+1.
\]
That is, these $m+1$ jobs are indistinguishable at the arrivals.
For a deterministic algorithm, the adversary sets the first two untested jobs, if exist,
or otherwise the last two jobs in the sequence to have their exact executing times $0$;
the other $m-1$ jobs have their exact executing times $2$.

For example, when $m = 4$ and the algorithm does not test any job, then the sequence of exact executing times is $\langle 0, 0, 2, 2, 2\rangle$;
if the algorithm does not test any of $J_3, J_4$ and $J_5$, then the sequence of exact executing times is $\langle 2, 2, 0, 0, 2\rangle$;
if the algorithm does not test $J_1$ only, then the sequence of exact executing times is $\langle 0, 2, 2, 2, 0\rangle$;
if the algorithm tests all the five jobs, then the sequence of exact executing times is $\langle 2, 2, 2, 0, 0\rangle$.
\end{instance}

Given a deterministic algorithm $A$, let $C^A_{\min}(I)$ denote the minimum machine load in the schedule produced by $A$ for the instance $I$.

\begin{lemma}
\label{lemma06}
When $m \ge 3$, on the instance $I$ in Instance~\ref{ex03}, for any deterministic algorithm $A$,
either $C^A(I) \ge 4$, or $C^A(I) = 3$ and $C_{\min}(I) \ge 2$.
\end{lemma}
\begin{proof}
If the algorithm $A$ does not test two or more jobs,
then one sees that the total processing time of every job is at least $2$, leading to the makespan $C^A(I) \ge 4$ by the Pigeonhole Principle.

If the algorithm $A$ tests all but at most one job,
then among the first $m$ jobs $\langle J_1, J_2, \cdots, J_m\rangle$, the total processing time of each but one of them is $3$,
and the total processing time of the exceptional job is either $1$ if it is tested or $2$ if it is untested.
One sees that when a machine is assigned with any two of these $m$ jobs, the makespan $C^A(I) \ge 4$ by the Pigeonhole Principle.
In the other case where each machine is assigned with exactly one of these $m$ jobs,
assigning the last job $J_{m+1}$ to any one of the machines would lead to either $C^A(I) \ge 4$, or $C^A(I) = 3$ and $C_{\min}(I) \ge 2$.
This finishes the proof.
\end{proof}

\begin{lemma}
\label{lemma07}
When $m = 2$, on the instance $I$ in Instance~\ref{ex03}, for any deterministic algorithm $A$,
either $C^A(I) \ge 5$,
or $C^A(I) \ge 4$ and $C_{\min}(I) \ge 1$,
or $C^A(I) \ge 3$ and $C_{\min}(I) \ge 2$.
\end{lemma}
\begin{proof}
Note that $I$ is now a three-job instance.

We distinguish the two cases on whether $A$ tests the first job $J_1$ or not.
If $J_1$ is tested, then $p_1 = 2$ and thus $p_1^A = 3$.
Note that the total processing times of $J_2$ and $J_3$ are at least $1$.
These three values $\{3, 1, 1\}$ guarantee one of the makespan scenarios.

If $J_1$ is untested, then $p_1^A = 2$.
Note that the total processing time of one of $J_2$ and $J_3$ is at least $2$ (which is either untested, or tested with the exact executing time $2$),
and of the other is at least $1$.
These three values $\{2, 2, 1\}$ guarantee one of the makespan scenarios.
\end{proof}

\begin{theorem}
\label{thm02}
The expected competitive ratio of any randomized algorithm for $P \mid online, t_j, 0 \le p_j \le u_j \mid C_{max}$
at the presence of three or more machines is at least $\frac {21}2 - \sqrt{78} \approx 1.6682$.
\end{theorem}
\begin{proof}
We consider a randomized instance $I_r$, which is a probability distribution over three deterministic instances $I_1, I_2$ and $I_3$.
These three instances are all extended from $I$ in Instance~\ref{ex03}, by appending a job $J_{m+2}$ at the end of the job sequence.
Specifically, for $J_{m+2}$,
\[
(u_{m+2}, t_{m+2}, p_{m+2}) = \left\{
\begin{array}{ll}
(0, 0, 0),									&\mbox{ in } I_1,\\
(2 + \frac 1\alpha, 3, 0),					&\mbox{ in } I_2,\\
(2 + \frac 1\alpha, 3, 2 + \frac 1\alpha),		&\mbox{ in } I_3,
\end{array}\right.
\]
where $\alpha = \frac {2 \sqrt{78} - 5}{41} \approx 0.3089$,
and its probability distribution is $(\beta_1, \beta_2, \beta_3) = (\frac {2 - \alpha}3, \alpha, \frac {1 - 2\alpha}3)$.

Note that exactly two of $\rho_1, \rho_2, \cdots, \rho_{m+1}$ are $1$, while the others are $2$.
For the job $J_{m+2}$, $\rho_{m+2} = 0, 3, 2 + \frac 1{\alpha}$ in $I_1, I_2, I_3$, respectively.
It follows from $m \ge 3$ that the optimal offline makespans of $I_1, I_2, I_3$ are $C^* = 2, 3, 2 + \frac 1{\alpha}$, respectively.
Therefore, the expected optimal offline makespan of $I_r$ is
\[
E[C^*(I_r)] = 2 \beta_1 + 3 \beta_2 + (2 + \frac 1\alpha) \beta_3 = \frac 43 + \alpha + \frac 1{3\alpha}.
\]

For any deterministic algorithm $A$, it performs the same on the prefix instance $I$ of all the three instances $I_1, I_2$ and $I_3$,
including which jobs of the first $m+1$ ones are tested, to which machines they are assigned to,
and all the $m$ machine loads when the last job $J_{m+2}$ arrives.
In particular, since $u_{m+2}$ and $t_{m+2}$ are the same in $I_2$ and $I_3$, $J_{m+2}$ is either both tested or both untested by $A$.
By Lemma~\ref{lemma06}, we have either $C^A(I_1) = C^A(I) \ge 4$, or $C^A(I_1) = C^A(I) = 3$ and $C^A_{\min}(I_1) = C^A_{\min}(I) \ge 2$.
We distinguish the following four cases on $C^A(I)$ and whether $J_{m+2}$ of $I_2$ is tested by $A$ or not.

The first case is that $C^A(I) \ge 4$ and $J_{m+2}$ of $I_2$ is tested by $A$.
Clearly, $C^A(I_2) \ge 4$ and $C^A(I_3) \ge p_{m+2}^A = 5 + \frac 1{\alpha}$.
Therefore,
\[
E[C^A(I_r)] \ge 4 \beta_1 + 4 \beta_2 + (5 + \frac 1{\alpha}) \beta_3 = \frac {11}3 - \frac {2 \alpha}3 + \frac 1{3 \alpha}.
\]

The second case is that $C^A(I) \ge 4$ and $J_{m+2}$ of $I_2$ is untested by $A$.
We have $C^A(I_i) \ge p_{m+2}^A = 2 + \frac 1{\alpha}$ for both $i = 2, 3$.
Therefore,
\[
E[C^A(I_r)] \ge 4 \beta_1 + (2 + \frac 1{\alpha})(\beta_2 + \beta_3) = \frac {11}3 - \frac {2 \alpha}3 + \frac 1{3 \alpha}.
\]

The third case is that $C^A(I) = 3$, $C^A_{\min}(I) \ge 2$, and $J_{m+2}$ of $I_2$ is tested by $A$.
For $I_2$, $p_{m+2}^A = 3$ and thus $C^A(I_2) \ge 5$ no matter which machine $J_{m+2}$ is assigned to;
for $I_3$, $p_{m+2}^A = 5 + \frac 1{\alpha}$ and thus $C^A(I_3) \ge 7 + \frac 1{\alpha}$.
Therefore,
\[
E[C^A(I_r)] \ge 3 \beta_1 + 5 \beta_2 + (7 + \frac 1{\alpha}) \beta_3 = \frac {11}3 - \frac {2 \alpha}3 + \frac 1{3 \alpha}.
\]

The last case is that $C^A(I) = 3$, $C_{\min}(I) \ge 2$, and $J_{m+2}$ of $I_2$ is untested by $A$.
For both $I_2$ and $I_3$, $p_{m+2}^A = 2 + \frac 1{\alpha}$ and thus $C^A(I_i) \ge 4 + \frac 1{\alpha}$ no matter which machine $J_{m+2}$ is assigned to,
for both $i = 2, 3$.
Therefore,
\[
E[C^A(I_r)] \ge 3 \beta_1 + (4 + \frac 1{\alpha})(\beta_2 + \beta_3) = \frac {11}3 + \frac {\alpha}3 + \frac 1{3 \alpha}
	> \frac {11}3 - \frac {2 \alpha}3 + \frac 1{3 \alpha}.
\]

To conclude, by Theorem~\ref{thm01},
the expected competitive ratio of any randomized algorithm for $P \mid online, t_j, 0 \le p_j \le u_j \mid C_{\max}$
at the presence of three or more machines is at least
\[
\inf_A \frac {E[C^A(I_r)]}{E[C^*(I_r)]} \ge \frac {\frac {11}3 - \frac {2 \alpha}3 + \frac 1{3 \alpha}}{\frac 43 + \alpha + \frac 1{3 \alpha}} 
	= \frac {21}2 - \sqrt{78}.
\]
This completes the proof of the theorem.
\end{proof}

\begin{theorem}
\label{thm03}
The expected competitive ratio of any randomized algorithm for $P2 \mid online, t_j, 0 \le p_j \le u_j \mid C_{max}$
is at least $\frac {21 + 4 \sqrt{51}}{30} \approx 1.6522$.
\end{theorem}
\begin{proof}
Again we consider a randomized instance $I_r$, which is a probability distribution over three deterministic instances $I_1, I_2$ and $I_3$.
These three instances are all extended from the three-job instance $I$ in Instance~\ref{ex03} when $m = 2$,
by appending a job $J_4$ at the end of the job sequence.
Specifically, for $J_4$,
\[
(u_4, t_4, p_4) = \left\{
\begin{array}{ll}
(0, 0, 0),									&\mbox{ in } I_1,\\
(\frac 2\alpha, 4, 0),					&\mbox{ in } I_2,\\
(\frac 2\alpha, 4, \frac 2\alpha),		&\mbox{ in } I_3,
\end{array}\right.
\]
where $\alpha = \frac {\sqrt{51} - 4}{10} \approx 0.3141$,
and its probability distribution is $(\beta_1, \beta_2, \beta_3) = (\frac 12, \alpha, \frac 12 - \alpha)$.

Recall that two of $\rho_1, \rho_2, \rho_3$ are $1$ while the other is $2$.
For the job $J_4$, $\rho_4 = 0, 4, \frac 2{\alpha}$ in $I_1, I_2, I_3$, respectively.
It follows that the optimal offline makespans of $I_1, I_2, I_3$ are $C^* = 2, 4, \frac 2{\alpha}$, respectively.
Therefore, the expected optimal offline makespan of $I_r$ is
\[
E[C^*(I_r)] = 2 \beta_1 + 4 \beta_2 + \frac 2\alpha \beta_3 = 4 \alpha + \frac 1\alpha - 1.
\]

For any deterministic algorithm $A$, it performs the same on the prefix instance $I$ of all the three instances $I_1, I_2$ and $I_3$,
including which of the first three jobs are tested, to which machines they are assigned to,
and the two machine loads when the last job $J_4$ arrives.
In particular, since $u_4$ and $t_4$ are the same in $I_2$ and $I_3$, $J_4$ is either both tested or both untested by $A$.
By Lemma~\ref{lemma07}, we have either $C^A(I_1) = C^A(I) \ge 5$, or $C^A(I_1) = C^A(I) \ge 4$ and $C^A_{\min}(I_1) = C^A_{\min}(I) \ge 1$,
or $C^A(I_1) = C^A(I) \ge 3$ and $C^A_{\min}(I_1) = C^A_{\min}(I) \ge 2$.
We distinguish the following six cases on $C^A(I)$ and whether $J_4$ of $I_2$ is tested by $A$ or not.

When $C^A(I) \ge 5$ and $J_4$ of $I_2$ is tested by $A$, we have $C^A(I_2) \ge 5$ and $C^A(I_3) \ge p_4^A = 4 + \frac 2\alpha$, leading to
\[
E[C^A(I_r)] \ge 5 \beta_1 + 5 \beta_2 + (4 + \frac 2\alpha)\beta_3 = \frac 52 + \alpha + \frac 1\alpha > \frac 52 + \frac 1\alpha.
\]

When $C^A(I) \ge 5$ and $J_4$ of $I_2$ is untested by $A$, we have $C^A(I_i) \ge p_4^A = \frac 2\alpha$ for both $i = 2, 3$, leading to
\[
E[C^A(I_r)] \ge 5 \beta_1 + \frac 2\alpha (\beta_2 + \beta_3) = \frac 52 + \frac 1\alpha.
\]

When $C^A(I) \ge 4$ and $C_{\min}(I) \ge 1$, and $J_4$ of $I_2$ is tested by $A$,
we have $p_4^A = 4$ in $I_2$ and thus $C^A(I_2) \ge p_4^A + 1 = 5$, and similarly $C^A(I_3) \ge p_4^A + 1 = 5 + \frac 2\alpha$.
They lead to
\[
E[C^A(I_r)] \ge 4 \beta_1 + 5 \beta_2 + (5 + \frac 2\alpha)\beta_3 = \frac 52 + \frac 1\alpha.
\]

When $C^A(I) \ge 4$ and $C_{\min}(I) \ge 1$, and $J_4$ of $I_2$ is untested by $A$,
we have $p_4^A = \frac 2\alpha$ in both $I_2$ and $I_3$ and thus $C^A(I_i) \ge p_4^A + 1 = 1 + \frac 2\alpha$ for both $i = 2, 3$.
They lead to
\[
E[C^A(I_r)] \ge 4 \beta_1 + (1 + \frac 2\alpha)(\beta_2+\beta_3) = \frac 52 + \frac 1\alpha.
\]

When $C^A(I) \ge 3$ and $C_{\min}(I) \ge 2$, and $J_4$ of $I_2$ is tested by $A$,
we have $p_4^A = 4$ in $I_2$ and thus $C^A(I_2) \ge p_4^A + 2 = 6$, and similarly $C^A(I_3) \ge p_4^A + 2 = 6 + \frac 2\alpha$.
They lead to
\[
E[C^A(I_r)] \ge 3 \beta_1 + 6 \beta_2 + (6 + \frac 2\alpha)\beta_3 = \frac 52 + \frac 1\alpha.
\]

Lastly, when $C^A(I) \ge 3$ and $C_{\min}(I) \ge 2$, and $J_4$ of $I_2$ is untested by $A$,
we have $p_4^A = \frac 2\alpha$ in both $I_2$ and $I_3$ and thus $C^A(I_i) \ge p_4^A + 2 = 2 + \frac 2\alpha$ for both $i = 2, 3$.
They lead to
\[
E[C^A(I_r)] \ge 3 \beta_1 + (2 + \frac 2\alpha)(\beta_2+\beta_3) = \frac 52 + \frac 1\alpha.
\]

To conclude, by Theorem~\ref{thm01},
the expected competitive ratio of any randomized algorithm for $P2 \mid online, t_j, 0 \le p_j \le u_j \mid C_{\max}$ is at least
\[
\inf_A \frac {E[C^A(I_r)]}{E[C^*(I_r)]} \ge \frac {\frac 52 + \frac 1\alpha}{4 \alpha + \frac 1\alpha - 1}
	= \frac {21 + 4\sqrt{51}}{30}.
\]
This completes the proof of the theorem.
\end{proof}

\subsection{On deterministic competitive ratios}
Recall that besides the lower bound of $2$ on the deterministic competitive ratio for the uniform testing case
$P \mid online, t_j = 1, 0 \le p_j \le u_j \mid C_{\max}$,
Albers and Eckl~\cite{AE21} also showed a slightly better lower bound of $2.0953$ for the two-machine general testing case
$P2 \mid online, t_j, 0 \le p_j \le u_j \mid C_{\max}$.
We improve the latter lower bound of $2.0953$ to $2.2117$ in the next theorem.
Therefore, for $P2 \mid online, t_j, 0 \le p_j \le u_j \mid C_{\max}$,
the expected competitive ratio of {\sc GCL} not only beats the best known deterministic competitive ratio of $2.3019$~\cite{GFL22}
for the semi-online variant $P2 \mid t_j, 0 \le p_j \le u_j \mid C_{\max}$,
but also beats the lower bound of $2.2117$ on the deterministic competitive ratio for the fully online problem.

\begin{theorem}
\label{thm06}
The competitive ratio of any deterministic algorithm for $P2 \mid online, t_j, 0\le p_j \le u_j \mid C_{max}$
is greater than $2.2117$.
\end{theorem}
\begin{proof}
Note that we are proving a lower bound on the deterministic competitive ratio,
and we use an adversarial argument to construct a three-job instance with the job order $\langle J_1, J_2, J_3\rangle$.
Consider a deterministic algorithm $A$.

The first job $J_1$ comes with $u_1 = \varphi$ and $t_1 = 1$.
If $A$ tests $J_1$, then the adversary sets $p_1 = \varphi$.
It follows that $p_1^A = \varphi + 1$ and $\rho_1 = \varphi$.
If $J_1$ is untested, then the adversary sets $p_1 = 0$ and thus $p_1^A = \varphi$ and $\rho_1 = 1$.
Either way, we have $\frac {p_1^A}{\rho_1} = \varphi$, and thus we assume below that $p_1^A = \varphi$ and $\rho_1 = 1$.
(If $p_1^A = \varphi + 1$ and $\rho_1 = \varphi$,
then the below $u_j$- and $t_j$-values associated with $J_2$ and $J_3$ are scaled up by multiplying a factor of $\varphi$.)

For any upcoming job $J_j$, i.e., $j = 2, 3$, the adversary always sets $p_j = u_j$ if $J_j$ is tested by the algorithm,
or otherwise sets $p_j = 0$.

Let $x_0 = \frac {3\varphi+1-\sqrt{11\varphi+6}}2 \approx 0.4878$, which is a root to the quadratic equation $x^2 - (3 \varphi + 1) x + \varphi^2 = 0$.
The second job $J_2$ comes with $u_2 = \varphi - x_0$ and $t_2 = x_0$.
We distinguish two cases on whether or not $J_2$ is tested by the algorithm.

\noindent{\bf Case 1:} $J_2$ is tested by the algorithm.
In this case, $p_2 = u_2$ and thus $p_2^A = \varphi$ and $\rho_2 = \varphi - x_0$.
If $J_2$ is scheduled on the same machine with $J_1$,
then the third job $J_3$ is voided (by setting $u_j = t_j = 0$) leading to $C^A = p_1^A + p_2^A = 2\varphi$.
Note that $\rho_2 > 1$ and thus $C^* = \rho_2 = \varphi - x_0$.
It follows that the competitive ratio is $\frac {2 \varphi}{\varphi - x_0}$.

Below we assume $J_j$ is assigned to the machine $M_j$, for $j = 1, 2$, respectively, and thus the load of each machine is $\varphi$.
The third job $J_3$ comes with $u_3 = y_0$ and $t_3 = \varphi + 1 - x_0$,
where $y_0 = \frac {1 - x_0 + \sqrt{(3 \varphi - 5) x_0 + 15 \varphi + 8}}2 \approx 3.0933$
is a root to the quadratic equation $y^2 - (1 - x_0) y - (\varphi + 1 - x_0)(2 \varphi + 1 - x_0) = 0$.

If $J_3$ is untested, then $p_3 = 0$ and thus $p_3^A = y_0$ and $\rho_3 = \varphi + 1 - x_0$, leading to $C^A = \varphi + y_0$.
In the optimal offline schedule, $J_1$ and $J_2$ are scheduled on one machine while $J_3$ is scheduled on the other,
leading to the optimal offline makespan $C^* = \varphi + 1 - x_0$.
Therefore, $\frac {C^A}{C^*} = \frac {y_0 + \varphi}{\varphi + 1 - x_0}$.

If $J_3$ is tested, then $p_3 = u_3$ and thus $p_3^A = t_3 + u_3 = \varphi + 1 - x_0 + y_0$ and $\rho_3 = y_0$,
leading to $C^A = 2 \varphi + 1 - x_0 + y_0$.
In the optimal offline schedule, $J_1$ and $J_2$ are scheduled on one machine while $J_3$ is scheduled on the other,
leading to $C^* = y_0$ and subsequently $\frac {C^A}{C^*} = \frac {2 \varphi + 1 - x_0 + y_0}{y_0}$.

To conclude this case, we have 
\[
\frac {C^A}{C^*} \ge \min\left\{ \frac {2 \varphi}{\varphi - x_0}, \frac {y_0 + \varphi}{\varphi + 1 - x_0}, \frac {2 \varphi + 1 - x_0 + y_0}{y_0} \right\}.
\]
Since $y_0$ is a root to the quadratic equation $y^2 - (1 - x_0) y - (\varphi + 1 - x_0)(2 \varphi + 1 - x_0) = 0$, we have
\[
y_0 (y_0 + \varphi) = (\varphi + 1 - x_0)(y_0 + 2\varphi + 1 - x_0),
\]
that is, the last two quantities in the above are equal and approximately $2.21172$.
The first quantity in the above is about $2.8634$.
Therefore, $\frac {C^A}{C^*} > 2.2117$.

\noindent{\bf Case 2:} $J_2$ is untested by the algorithm.
In this case, $p_2 = 0$ and thus $p_2^A = \varphi - x_0$ and $\rho_2 = x_0$.
If $J_2$ is scheduled on the same machine with $J_1$,
then the third job $J_3$ is voided (by setting $u_j = t_j = 0$) leading to $C^A = p_1^A + p_2^A = 2\varphi - x_0$.
Note that $\rho_2 = x_0 < 1$ and thus $C^* = 1$.
It follows that the competitive ratio is $\frac {C^A}{C^*} = 2 \varphi - x_0$.

Below we assume $J_j$ is assigned to the machine $M_j$, for $j = 1, 2$, respectively,
and thus the loads of the two machines are $\varphi$ and $\varphi - x_0$, respectively.
The third job $J_3$ comes with $u_3 = \frac {(1+\varphi)y_0}{2\varphi+1-x_0}\approx 2.1606$ and $t_3 = 1 + x_0$,
where $y_0$ is set the same as in Case 1.

If $J_3$ is untested, then $p_3 = 0$ and thus $p_3^A = u_3 = \frac {(1 + \varphi) y_0}{2 \varphi + 1 - x_0}$ and $\rho_3 = 1 + x_0$,
leading to $C^A \ge \varphi - x_0 + \frac {(1 + \varphi) y_0}{2 \varphi + 1 - x_0}$ no matter which machine $J_3$ is assigned to.
In the optimal offline schedule, $J_1$ and $J_2$ are scheduled on one machine while $J_3$ is scheduled on the other,
leading to the optimal offline makespan $C^* = 1 + x_0$.
Therefore, $\frac {C^A}{C^*} \ge \frac {\varphi - x_0 + \frac {(1 + \varphi) y_0}{2 \varphi + 1 - x_0}}{1 + x_0}
= \frac {(\varphi - x_0)(2 \varphi + 1 - x_0) + (1 + \varphi) y_0}{(1 + x_0)(2 \varphi + 1 - x_0)}$.
Since $x_0$ is a root to the quadratic equation $x^2 - (3 \varphi + 1) x + \varphi^2 = 0$, we have 
\[
(\varphi - x_0)(2 \varphi + 1 - x_0) = x_0^2 - (3 \varphi + 1) x_0 + \varphi (1 + 2 \varphi) = \varphi (1 + \varphi)
\]
and
\[
(1 + x_0)(2 \varphi + 1 - x_0) = -x_0^2 + 2 \varphi x_0 + 2 \varphi + 1 = (\varphi + 1)(1 + \varphi - x_0).
\]
It follows that $\frac {C^A}{C^*} \ge \frac {y_0 + \varphi}{1 + \varphi - x_0}$.

If $J_3$ is tested, then $p_3 = u_3$ and thus $p_3^A = t_3 + u_3 = 1 + x_0 + \frac {(1 + \varphi) y_0}{2 \varphi + 1 - x_0}$
and $\rho_3 = \frac {(1 + \varphi) y_0}{2 \varphi + 1 - x_0}$,
leading to $C^A \ge 1 + \varphi + \frac {(1 + \varphi) y_0}{2 \varphi + 1 - x_0}$ no matter which machine $J_3$ is assigned to.
Since $\rho_3 > \rho_1 + \rho_2$, in the optimal offline schedule,
$J_1$ and $J_2$ are scheduled on one machine while $J_3$ is scheduled on the other,
leading to $C^* = \rho_3 = \frac {(1 + \varphi) y_0}{2 \varphi + 1 - x_0}$ and subsequently
$\frac {C^A}{C^*} \ge \frac {y_0 + 2 \varphi + 1 - x_0}{y_0}$.

To conclude this case, we have 
\[
\frac {C^A}{C^*} \ge \min\left\{ 2 \varphi - x_0, \frac {y_0 + \varphi}{1 + \varphi - x_0}, \frac {2 \varphi + 1 - x_0 + y_0}{y_0} \right\}.
\]
The same as in Case 1, the last two quantities in the above are equal and approximately $2.21172$.
The first quantity in the above is about $2.7482$.
Therefore, $\frac {C^A}{C^*} > 2.2117$.

The above two cases prove that the competitive ratio of any deterministic algorithm for
$P2 \mid online, t_j, 0\le p_j \le u_j \mid C_{max}$ is greater than $2.2117$.
\end{proof}

\section{Conclusion}
We investigated the fully online multiprocessor scheduling with testing problem $P \mid online, t_j, 0 \le p_j \le u_j \mid C_{\max}$,
and presented a randomized algorithm {\sc GCL} that is a non-uniform distribution of arbitrarily many deterministic algorithms.
To the best of our knowledge, randomized algorithms as meta algorithms in the literature are mostly uniform distributions of their
component deterministic algorithms, while our {\sc GCL} is one to bias towards one component algorithm;
also, when there are many component algorithms, the previous randomized algorithm often involves bookkeeping all their solutions,
while our {\sc GCL} does not do so since the component algorithms are independent of each other.
We showed that the expected competitive ratio of our {\sc GCL} is around $3.1490$.
When there are only two machines, i.e., for $P2 \mid online, t_j, 0 \le p_j \le u_j \mid C_{\max}$,
using only two component algorithms in the revised {\sc GCL} algorithm leads to an expected competitive ratio of $2.1839$.

We also proved three inapproximability results,
including a lower bound of $1.6682$ on the expected competitive ratio of any randomized algorithm at the presence of at least three machines,
a lower bound of $1.6522$ on the expected competitive ratio of any randomized algorithm for
$P2 \mid online, t_j, 0 \le p_j \le u_j \mid C_{\max}$,
and a lower bound of $2.2117$ on the competitive ratio of any deterministic algorithm for $P2 \mid online, t_j, 0 \le p_j \le u_j \mid C_{\max}$.
By the last lower bound, we conclude that the algorithm {\sc GCL} beats any deterministic algorithm for $P2 \mid online, t_j, 0 \le p_j \le u_j \mid C_{\max}$,
in terms of its expected competitive ratio.
Such an algorithmic result is rarely seen in the literature.

The expected competitive ratio of the algorithm {\sc GCL} is far away from the lower bounds of $1.6522$, $1.6682$ and $2 - \frac 1m$
for two, three and $m \ge 4$ machines, respectively, suggesting future research to narrow the gaps, even for the two-machine case.
One sees that the testing probability function in the job ratio in the algorithm {\sc GCL} is non-continuous at $\varphi$,
suggesting the existence of a plausible better probability distribution function for choosing component algorithms.



\section*{Declarations}
\paragraph*{Data availability.}
Not applicable.

\paragraph*{Interests.}
The authors declare that they have no known competing financial interests or personal relationships
that could have appeared to influence the work reported in this paper.


\end{document}